%% file: main.tex
\documentclass[11pt]{article}

\usepackage{blindtext}

\usepackage[abbrvbib, preprint]{jmlr2e}

\usepackage[utf8]{inputenc} 
\usepackage[T1]{fontenc}    
\usepackage{hyperref}       
\usepackage{url}            
\usepackage{booktabs}       
\usepackage{amsfonts}       
\usepackage{nicefrac}       
\usepackage{microtype}      
\usepackage{xcolor}         
\usepackage[english]{babel}

\usepackage{algorithm, algpseudocode}

\usepackage{amsmath, 
            amssymb, 
            mathtools, 
            enumitem, 
            float, 
            subcaption, 
            multirow, 
            booktabs, 
            tcolorbox}
\usepackage{tikz}
\usetikzlibrary{arrows.meta,patterns,calc,positioning}

\usepackage{placeins}

\newcommand{\dataset}{{\cal D}}

\newcommand{\R}{\mathbb{R}}

\newtheorem{assumption}{Assumption}

\usepackage{lastpage}
\jmlrheading{23}{2025}{1-\pageref{LastPage}}{12/25; Revised 12/25}{12/25}{21-0000}{Yunjie Fan and Matteo Sesia}

\ShortHeadings{Interpretable Multivariate Conformal Prediction}{Fan and Sesia}
\firstpageno{1}

\begin{document}

\title{Interpretable Multivariate Conformal Prediction with Fast Transductive Standardization}

\author{\name Yunjie Fan \email brianfan@usc.edu \\
       \addr Department of Mathematics\\
       University of Southern California\\
       Los Angeles, CA, USA
       \AND
       \name Matteo Sesia \email sesia@marshall.usc.edu \\
       \addr Department of Data Sciences and Operations\\
       {Thomas Lord Department of Computer Science\\
       University of Southern California\\
       Los Angeles, CA, USA}
}
\editor{My editor}

\maketitle

\begin{abstract}
\input{abstract.tex}
\end{abstract}


\input{body}

\bibliography{biblio} 

\newpage

\appendix
\include{supplementary}

\end{document}

%% file: abstract.tex
We propose a conformal prediction method for constructing tight simultaneous prediction intervals for multiple, potentially related, numerical outputs given a single input. This method can be combined with any multi-target regression model and guarantees finite-sample coverage. It is computationally efficient and yields informative prediction intervals even with limited data. The core idea is a novel \emph{coordinate-wise} standardization procedure that makes residuals across output dimensions directly comparable, estimating suitable scaling parameters using the calibration data themselves. This does not require modeling of cross-output dependence nor auxiliary sample splitting. Implementing this idea requires overcoming technical challenges associated with transductive or full conformal prediction. Experiments on simulated and real data demonstrate this method can produce tighter prediction intervals than existing baselines while maintaining valid simultaneous coverage.


%% file: body.tex
\section{Introduction}
\label{sec:intro}

\subsection{Motivation and Preview of Contributions}
\label{sec:intro-motivation} 

We are interested in jointly quantifying uncertainty when predicting several, possibly dependent continuous outcomes from a shared set of input features. 
This problem is motivated by many applications where data-driven decisions depend on several unknown variables rather than on any one of them in isolation. For example, a financial institution may use available customer data to jointly forecast that customer's future credit-line utilization, monthly spending volume, and credit score. In this and many other contexts, decision quality hinges on the joint reliability of all predicted values, and an output that is incorrect in even one dimension could lead to adverse consequences.

Conformal prediction is a popular and broadly applicable framework for uncertainty quantification in predictive settings \citep{vovk2005algorithmic, angelopoulos2023conformal}. It derives its main strengths from the ability to augment the output of any machine-learning model with {\em prediction sets} equipped with finite-sample guarantees, without relying on strong assumptions on the data distribution. Its sole assumption is the availability of labeled data that are {\em exchangeable} with the test point of interest, a weaker requirement than the standard assumption of i.i.d.~sampling. For a {\em univariate} setting, where the goal is to predict a single continuous-valued outcome, there now exist several well-established conformal prediction methods that produce valid, computationally fast, and easy-to-interpret prediction intervals \citep{lei2018distribution, romano2019conformalized, sesia2021conformal}.

Extending conformal prediction to the {\em multivariate} setting is a non-trivial problem that has attracted considerable recent interest, spurring a variety of solutions \citep{dheur2025unified}. The fundamental challenge is to aggregate a vector of generalized residuals—whose components may differ substantially in scale and variability—into a scalar-valued {\em non-conformity score} leading to {\em prediction sets} that should satisfy four key desiderata: (i) finite-sample coverage, (ii) informativeness or tightness, (iii) computational tractability, and (iv) interpretability.
Interpretability can take different meanings depending on the application, and in some settings it is consistent with prediction sets having complex shapes. Nevertheless, in many practical scenarios, hyper-rectangular prediction sets—or, equivalently, jointly valid prediction intervals—provide the most natural and transparent way to convey uncertainty to practitioners. This paper focuses on such settings.

Existing methods achieve only a subset of these desiderata. The simplest approach to obtain hyper-rectangular prediction sets is to apply univariate conformal prediction to each output and then take the Cartesian product of the resulting intervals \citep{neeven2018conformal}; however, a conservative Bonferroni adjustment to the marginal miscoverage levels is needed to guarantee joint coverage; see Algorithm~\ref{alg:bonferroni} in Appendix~\ref{app:algorithms}. This strategy is very easy to implement, but the Bonferroni correction is often too conservative to be practically satisfactory, as it adopts a worst-case view of the dependence structure across output dimensions.
By contrast, we aim to achieve valid coverage with joint prediction intervals that are as tight as possible; see Figure~\ref{fig:joint-prediction} for a sketch of this goal.

\begin{figure}[!htb]
    \centering
    \includegraphics[width = 0.9\textwidth]{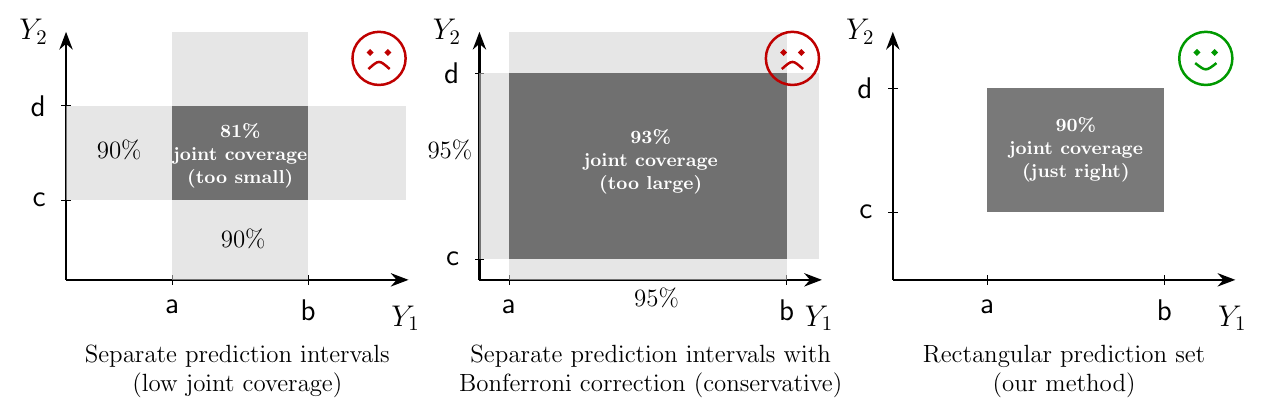}
\caption{Two-dimensional toy example depicting graphically the main objective of this work: constructing \emph{informative} and \emph{interpretable} hyper-rectangular conformal prediction sets for multivariate outcomes. 
Left: univariate approaches can produce marginal prediction intervals for each output separately, which do not guarantee joint coverage. 
Center: a Bonferroni correction can easily restore valid joint coverage, at the cost of being overly conservative. 
Right: our goal is to efficiently construct hyper-rectangular prediction sets that are easy to interpret and retain valid coverage; doing so involves technical challenges.}
    \label{fig:joint-prediction}
\end{figure}

A second class of existing methods reduce multivariate residuals to scalar-valued non-conformity scores, for example by applying an $\ell_2$ or $\ell_\infty$ norm, then apply a univariate conformal prediction method to these scalar scores, and finally map the resulting prediction interval back into the original output space by inverting the dimension-reduction transformation; e.g., see Algorithm~\ref{alg:unscaled} in Appendix~\ref{app:algorithms}. While computationally efficient and in principle able to produce hyper-rectangular prediction sets (e.g., using the $\ell_\infty$-norm), these approaches are dominated by the noisiest dimension, often leading to prediction sets that are unnecessarily wide along low-uncertainty directions.

More sophisticated approaches have been developed to directly minimize the volume of the resulting prediction sets; for example using optimal transport, learned embeddings, copula, or generative models to describe the dependence structure across residual dimensions. However, these methods aim to solve a problem that is different from that considered in this paper, since they are explicitly designed to produce non-rectangular prediction sets. Moreover, they often introduce substantial computational overhead and require additional independent training data to provide finite-sample guarantees.

In this paper we focus on targeting what we believe is the core difficulty in many applications: the mismatch in the {\em marginal} scale and variability of residuals across output dimensions. 
To address this challenge, we develop a method based on transductive coordinate-wise residual standardization. This restores comparability across coordinates and enables the construction of tight hyper-rectangular prediction sets that guarantee finite-sample coverage, without necessitating additional training data or prohibitive computational overhead. Remarkably, our method does {\em not} need to explicitly model residual dependencies across different dimensions, unlike existing approaches focusing on non-rectangular prediction sets.

Making this idea operational requires overcoming several hurdles, which we resolve through a sequence of increasingly practical steps. We begin by formulating an ideal oracle procedure that assumes access to the test point’s true outcome vector and uses this information to standardize the multi-dimensional residuals in the calibration set without violating the exchangeability conditions that are necessary to guarantee finite-sample coverage. Although clearly impractical, this oracle is valuable to guide the design of our practical method.
Inspired by the oracle procedure, we introduce two increasingly sophisticated data-driven approaches. The first provides a more conservative approximation of the oracle, whereas the second achieves greater statistical efficiency by building on top of the first. Through extensive numerical experiments, we demonstrate that this progression of ideas culminates in a practical method producing prediction sets satisfying all our desiderata.

\subsection{Related Works}

Extending conformal prediction to multivariate outcomes has attracted substantial attention in recent years, leading to the development of a wide range of alternative methods.

A prominent line of work seeks minimum-volume prediction sets with flexible geometries by modeling dependence across output dimensions.
This includes copula-based approaches \citep{messoudi2021copula, pmlr-v258-park25c, sun2024copula}, non-parametric density estimation methods \citep{sadinle2016least,izbicki2022cdsplit,wang2023pcp, dheur2025unified}, parametric density estimation methods \citep{johnstone2021conformal,sampson2024conformal,braun2025minimum,braun2025multivariate},  optimal-transport techniques \citep{thurin2025otcp, klein2025multivariate}, and flow-based generative models \citep{colombo2024normalizing, lee2025flow, fang2025contra}. While these approaches can produce smaller prediction sets compared to our method, these often have complex shapes that are more difficult to interpret. Moreover, these approaches require additional model training on independent data to provide finite-sample guarantees, and many also involve substantial computational overhead.

To the best of our knowledge, the idea of directly addressing heterogeneity across output dimensions via marginal residual rescaling remains relatively under-explored. \citet{baheri2025multiscale} pursue a related idea, but rely on a conservative Bonferroni-type correction. The approach of \citet{sampson2024conformal} is more closely related to ours, but requires an additional independent dataset; avoiding this requirement is the main technical challenge addressed in our paper. As we show empirically, the method of \citet{baheri2025multiscale} is often considerably less efficient than ours—particularly in small-sample regimes—resulting in less stable and less informative prediction sets.

Our solution can be viewed as implementing a special instance of transductive, or full, conformal prediction \citep{vovk2005algorithmic,vovk2013transductive}. In contrast to the more commonly used split conformal methods, full conformal prediction repeatedly incorporates the test point, augmented with different candidate labels, directly into a learning algorithm. Although this strategy is computationally prohibitive in general, especially for continuous-valued numerical outcomes and complex algorithms, it can be implemented efficiently in some cases by exploiting structural properties of the underlying predictive model. Accordingly, our paper relates to a broader literature on making full conformal inference computationally practical in special cases, including approaches tailored to sparse generalized linear models \citep{lei2019fast,guha2023conformalization} and methods that leverage stability or other forms of structure in the predictive model \citep{ndiaye2019computing, abad2022approximating, martinez2023approximating, li2024generalized}. Unlike those works, which focus primarily on accelerating model refitting, we assume a fixed, pre-trained model and instead aim to improve efficiency by avoiding calibration data splitting while explicitly standardizing residuals across output dimensions. Both the problem setting and the proposed methodology are therefore novel.

More broadly, our approach is also related to recent conformal methods that reuse calibration data for different purposes, such as enabling early stopping to control overfitting when training deep learning models \citep{liang2023conformal} and leveraging data-adaptive grouping to improve conditional coverage \citep{zhou2024conformal, zhang2024posterior}. In contrast to these works, our focus is on addressing heterogeneity across targets in multivariate regression while obtaining interpretable, rectangular prediction sets.

\subsection{Outline}

This paper is organized as follows.
Section~\ref{sec:preliminaries} introduces notations, a formal problem statement, and necessary assumptions, along with a description of an ideal oracle approach.
Section~\ref{sec:methods} presents our methodology, starting with a more conservative approach and then developing more efficient and computationally tractable implementations.
Section~\ref{sec:experiments} validates our method empirically and compares its performance to that of existing approaches using both simulated and real data.
Section~\ref{sec:discussion} discusses some limitations and possible directions for future research.
The Appendices contain mathematical proofs not contained in the main text, additional algorithmic details, and further numerical results.
 
\section{Preliminaries}
\label{sec:preliminaries}

\subsection{Notation}

For any $n \in \mathbb{N}$, we denote $[n] := \{1,\ldots,n\}$.
For $\alpha \in (0,1)$ and any $x=(x_1,\ldots,x_n) \in \mathbb{R}^n$, we denote $\hat{Q}_{1-\alpha}(x) := \lceil (1-\alpha)(n+1)\rceil\text{-th smallest value in } \{x_1,\ldots,x_n, +\infty\}$.

\subsection{Setup and Data Assumptions}

We consider a multi-output (multivariate) regression setting where the goal is to predict a vector $Y = (Y_1, \ldots, Y_d)$ of $d$ real-valued outputs given input features $X \in \mathcal{X}$. 
We assume access to a \emph{calibration} dataset of $n > 1$ labeled examples $\dataset_{\mathrm{cal}} = \{(X^i, Y^i)\}_{i=1}^n$, drawn exchangeably from an unknown joint distribution $P_{X,Y} = P_X \times P_{Y|X}$.

For a new independent data point $(X^{n+1}, Y^{n+1}) \sim P_{X,Y}$ of which only the input $X^{n+1}$ is observed, the goal is to construct a reasonably tight prediction set $\hat{C}(X^{n+1}) \subset \mathbb{R}^d$ that contains the output $Y^{n+1}$ with high probability, guaranteeing \emph{marginal coverage}:
\begin{equation}
\mathbb{P}\!\left[ Y^{n+1} \in \hat{C}(X^{n+1}) \right] \ge 1 - \alpha,
\label{eq:joint-coverage}
\end{equation}
where $\alpha \in (0,1)$ is the user-specified miscoverage level.  
Note that both $Y^{n+1}$ and $\hat{C}(X^{n+1})$ are random in~\eqref{eq:joint-coverage}, as they depend on the new test point as well as on the calibration data $\mathcal{D}_{\mathrm{cal}}$.  
The term ``marginal coverage’’ reflects that the guarantee in~\eqref{eq:joint-coverage} holds on average over the distribution of test points, rather than conditioning on any fixed value of $X^{n+1}$.

While this paper focuses on achieving marginal coverage because it is a standard goal in conformal prediction, the proposed method could be extended to achieve stronger coverage guarantees. In particular, full \emph{conditional coverage} is known to be generally unattainable in finite samples without much more restrictive assumptions \citep{foygel2021limits}, but there are many techniques to approximate it \citep{Romano2020With,gibbs2025conformal} that may be possible to combine with our approach.

For interpretability, we aim to construct \emph{hyper-rectangular} prediction sets that can be written as a Cartesian product of univariate prediction  intervals:
\begin{equation} \label{eq:rectangle}
\hat{C}(X^{n+1}; \mathcal{D}_{\mathrm{cal}}) 
= \bigtimes_{j=1}^d 
\bigl[ L_j(X^{n+1}; \mathcal{D}_{\mathrm{cal}}), 
       U_j(X^{n+1}; \mathcal{D}_{\mathrm{cal}}) \bigr],
     \end{equation}
where $[L_j(X^{n+1}), U_j(X^{n+1})]$ denotes the marginal interval for the $j$th output.  
Prediction sets in this form decompose naturally across dimensions, making them very easy to explain and interpret.  
In particular, under~\eqref{eq:rectangle}, the event $Y^{n+1} \in \hat{C}(X^{n+1})$ in~\eqref{eq:joint-coverage} is equivalent to
\[
Y^{n+1}_j \in [L_j(X^{n+1}), U_j(X^{n+1})]
\quad \text{for all } j \in [d],
\]
which means all $d$ marginal prediction intervals $[L_j(X^{n+1}), U_j(X^{n+1})]$ must {\em simultaneously} cover their respective outputs $Y^{n+1}_j$ with probability at least $1 - \alpha$.

\subsection{Computing Residuals using a Pre-Trained Model} \label{sec:residuals}

We assume access to a \emph{pre-trained} joint predictive model $\hat{f} = (\hat{f}_1,\ldots,\hat{f}_d)$, fitted on data independent of both the calibration set $\mathcal{D}_{\mathrm{cal}}$ and the test point. 
This model is treated as a ``black-box'' and could be anything; the only requirement for $\hat{f}$ is that it must provide us with a coordinate-wise \emph{generalized residual} function
\[
E : \mathcal{X} \times \mathbb{R}^d \to \mathbb{R}^d,
\]
which aims to quantify the model's prediction error along each output dimension.
Applying this function to a data point $(X,Y)$ produces the vector $E(X,Y) = \bigl(E_1(X,Y), \ldots, E_d(X,Y)\bigr)$.

A classical example for models $\hat{f}_j$ that estimate the conditional mean of $Y_j$ given $X$ \citep{lei2018distribution} is $E_j(X,Y) = \lvert Y_j - \hat{f}_j(X) \rvert$.
Alternatively, for models that estimate conditional quantiles, as in conformalized quantile regression \citep{romano2020classification}, a natural choice is
\begin{align*}
E_j(X,Y) = \max \bigl\{\hat{f}_j^{\mathrm{lo}}(X) - Y_j,\; Y_j - \hat{f}_j^{\mathrm{hi}}(X)\bigr\},
\end{align*}
where $\hat{f}_j^{\mathrm{lo}}(X)$ and $\hat{f}_j^{\mathrm{hi}}(X)$ denote the estimated lower and upper conditional quantiles of $Y_j$ given $X$, for example at levels $\alpha/2$ and $1-\alpha/2$. Intuitively, this choice of $E$ measures the signed distance of $Y_j$ to the nearest boundary of the  interval $[\hat{f}_j^{\mathrm{lo}}(X), \hat{f}_j^{\mathrm{hi}}(X)]$.
Positive values indicate that $Y_j$ falls outside the interval, while negative values indicate that it lies inside.

Without much loss of generality, we assume throughout the paper that the residuals are almost surely {\em distinct} and {\em non-negative}---which implies $E_j(X,Y) > 0$ for all $j$---since this simplifies the development of our method. Moreover, we also assume these residuals have finite first and second moments at the population level.

\begin{assumption} \label{eq:assumption-scores}
For all $j \in [d]$, the distribution of $E_j(X,Y)$ has no point masses and satisfies: $E_j(X,Y) \geq 0$ almost surely, $ \mathbb{E}[E_j(X,Y)] < \infty$, and $0 < \operatorname{Var}[E_j(X,Y)] < \infty$.
\end{assumption}

The assumption of no point masses is standard in conformal inference and practically without loss of generality, since any ties can always be broken by introducing a small amount of independent random noise. The bounded second moments are also immediately satisfied by any residuals whose distribution has finite support, which is common in practice.
The non-negativity assumption is satisfied by absolute residuals and many other standard choices, although it does not hold for the quantile-based residuals reviewed above. Nonetheless, there are at least two simple ways to adapt conformalized quantile regression to this assumption.
One is to simply truncate negative residuals at zero, which tends to make the corresponding conformal prediction sets more conservative \citep{romano2020classification}.
Alternatively, provided that the residuals are almost-surely bounded, as it is often the case in practice, non-negativity can be enforced by applying a sufficiently large constant shift to all calibration and test residuals---a transformation that can then be easily inverted to keep the resulting prediction sets invariant.
This idea will become clearer after we review a general recipe for translating residuals into conformal prediction sets.

\subsection{Calibrating Prediction Sets with Marginal Coverage}

After evaluating the residuals $E^i = E(X^i, Y^i)$ for all calibration points indexed by $i \in [n]$, a prediction set for $Y_{n+1}$ based on $X_{n+1}$ can be assembled by including all possible values of $y \in \mathbb{R}^d$ whose residuals $(E_1(X^{n+1},y), \ldots, E_d(X^{n+1},y))$ are sufficiently small:
\begin{equation}
    \hat{C}(X^{n+1}) = \left\{y \in \R^d: 0\leq E_j(X^{n+1},y) \leq W_j, \quad \forall j \in [d] \right\},
    \label{eq:joint-prediction}
\end{equation}
where $W_1,\ldots,W_d$ are suitable (data-dependent) thresholds.
To ensure that $\hat{C}(X^{n+1})$ satisfies marginal coverage at level $1-\alpha$, as defined in~\eqref{eq:joint-coverage}, the thresholds must be chosen such that
\begin{align*}
\mathbb{P} \left[ E_j(X^{n+1}, Y^{n+1}) \leq W_j, \; \forall j \in [d] \right] \geq 1-\alpha.
\end{align*}

A naive way to achieve this is to set $W_j = \hat{Q}_{1-\alpha/d}(E^1_j, \ldots ,E^n_j)$, which is equivalent to applying univariate conformal prediction separately across all $d$ outcome dimensions, with a Bonferroni correction ($\alpha/d$ replacing $\alpha$).
As demonstrated empirically in Appendix~\ref{app:additional-experiments}, this typically leads to very conservative prediction sets because it takes an overly pessimistic worst-case view of possible dependencies of the prediction tasks across outcome dimensions.

In this paper, we pursue a more efficient approach, which starts by computing {\em scalar} non-conformity scores $S^i = \Phi(E^i) \in \mathbb{R}$ for all $i \in [n]$ using a dimension-reduction function $\Phi : \mathbb{R}^d_+ \mapsto \mathbb{R}$, whose form will be discussed below.
Then, the prediction set is constructed as:
\begin{equation}
    \hat{C}(X^{n+1}) = \left\{y \in \R^d: \Phi(E(X^{n+1},y)) \leq \hat{Q}_{1-\alpha}(S^1,\ldots,S^n) \right\}.
    \label{eq:joint-prediction-scalar}
\end{equation}

\begin{theorem}[\citet{vovk2005algorithmic}] \label{thm:fixed-transformation-coverage}
Assume $(E^1,\ldots,E^{n+1})$ are exchangeable, with $E^i = E(X^i, Y^i)$ for all $i \in [n+1]$.
Let $\Phi:\R^d_+ \mapsto \R$ be any \emph{fixed} function, independent of the data.
Then, $\hat{C}(X^{n+1})$ defined in~\eqref{eq:joint-prediction-scalar} satisfies $\mathbb{P}[ Y^{n+1} \in \hat{C}(X^{n+1}) ] \ge 1 - \alpha$.
Moreover, if the scores $S^1,\ldots,S^n$ given by $S^i = \Phi(E^i)$ are almost-surely distinct, $\mathbb{P}[ Y^{n+1} \in \hat{C}(X^{n+1}) ] \leq 1 - \alpha + \frac{1}{n+1}$.
The same results hold with $\Phi$ random, if $(S^1,\ldots,S^{n+1})$ are exchangeable conditional on it.
\end{theorem}

Therefore, under Assumption~\ref{eq:assumption-scores}, Theorem~\ref{thm:fixed-transformation-coverage} guarantees the prediction sets defined in~\eqref{eq:joint-prediction-scalar} achieve marginal coverage {\em tightly} from above, as long as the sample size $n$ is not too small.
A limitation of Theorem~\ref{thm:fixed-transformation-coverage}, however, is that it cannot tell us whether the coverage is equally tight across different outcome dimensions.
In fact, depending on the data distribution, the fitted model, and the choice of $\Phi$, it could happen that $\hat{C}(X^{n+1})$ achieves the desired marginal coverage while being much wider than necessary along certain directions, thereby providing uncertainty estimates that are potentially much less informative than ideal.
Moreover, depending on $\Phi$, the prediction set $\hat{C}(X^{n+1})$ may have a difficult-to-interpret shape.

Several recent works propose residual transformations based on $\ell_p$-norms, i.e., $\Phi(\cdot) = \|\cdot\|_p$ for different $1\leq p \leq \infty$ \citep{johnstone2022exact}. Although conceptually and computationally simple, these can be heavily affected by the nosiest output dimension. A single high-variance dimension may dominate the scalar score and inflate the resulting prediction set uniformly across all coordinates. 
To obtain tighter prediction sets, recent works have proposed learning a suitable transformation of the residuals prior to applying a norm-based dimension reduction, for example by using optimal transport \citep{thurin2025otcp, klein2025multivariate}.
However, these approaches require additional training data as well as expensive computations, and they typically lead to prediction sets with irregular shapes.

In this paper, we propose a conceptually simpler solution that consists of {\em standardizing} the residuals prior to reducing their dimension by taking the maximum across coordinates, similar to the $\ell_\infty$ norm. This leads to rectangular prediction sets that are uniformly tight across all output dimensions, without requiring additional training data or expensive computations.
Despite the apparent simplicity of this idea, however, achieving all desiderata while maintaining guaranteed finite-sample coverage involves significant technical challenges, which we overcome as gradually explained below.

\subsection{An Ideal Population Oracle}
\label{sec:pop-oracle}
 
To motivate our method, we begin by describing an idealized {\em oracle} approach that assumes access to population-level information not available in practice. This oracle highlights the structure of the prediction sets we ultimately aim to approximate and achieves our four key desiderata: validity, tightness, computational tractability, and interpretability.

Consider prediction sets of the form in~\eqref{eq:joint-prediction-scalar}, obtained by transforming the vector-valued residuals $E(X,Y) \in \mathbb{R}_+^d$ using a suitable function $\Phi : \mathbb{R}^{d}_+ \mapsto \mathbb{R}$.
To address the possible heterogeneity in scale and variability across different dimensions, a natural idea is to standardize them one by one, which suggests using a function $\Phi$ defined as:
\begin{equation}
\Phi(t; \mu, \sigma)
:= \max_{1 \le j \le d} \frac{t_j - \mu_j}{\sigma_j}, \quad \forall t \in \R^d_+.
\label{eq:residual-transformation}
\end{equation}
Above, $\mu, \sigma \in \mathbb{R}^d_+$ are a location and scale parameter, respectively, which should ideally be equal to the population mean and standard deviation of the residual vector:
\begin{align} \label{eq:oracle-mean-std}
&\mu^*_j = \mathbb{E}[E_j(X,Y)], 
&\sigma^*_j = \sqrt{\operatorname{Var}[E_j(X,Y)]},
&&\forall j \in [d].
\end{align}
With this choice of $\Phi$, the prediction sets defined in~\eqref{eq:joint-prediction-scalar} become:
\begin{equation}
    C^{\mathrm{pop}}(X^{n+1}) = \left\{y \in \R^d: 0\leq E_j(X^{n+1}, y) \leq \hat{Q}^{\mathrm{pop}}_{1-\alpha} \cdot \sigma^*_j + \mu^*_j, \quad \forall j \in [d]\right\},
    \label{eq:pop-oracle}
\end{equation}
where $\hat{Q}^{\mathrm{pop}}_{1-\alpha} := \hat{Q}_{1-\alpha}(S^1_{\mathrm{pop}},\ldots,S^n_{\mathrm{pop}})$ and $S^i_{\mathrm{pop}}:= \Phi(E^i; \mu^*, \sigma^*)$ for all $i \in [n]$.

Although practically unfeasible, this oracle prediction set satisfies all four of our desiderata: (i) it guarantees finite-sample marginal coverage, as per Theorem~\ref{thm:fixed-transformation-coverage}; (ii) it tends to be uniformly tight across all output dimensions, because it operates with residuals whose components are all on the same scale; (iii) it is fast to compute, at least for standard choices of the residual function $E$; (iv) it is an easy-to-interpret hyper-rectangle, at least for standard choices of $E$.

To see that prediction sets $\hat{C}(X^{n+1})$ in the form of~\eqref{eq:joint-prediction}, of which~\eqref{eq:pop-oracle} is a special instance corresponding to $W_j = \hat{Q}^{\mathrm{pop}}_{1-\alpha} \cdot \sigma^*_j + \mu^*_j$, are easy-to-compute rectangles, consider the case of absolute mean-regression residuals, i.e., $E_j(X,Y) = \lvert Y_j - \hat{f}_j(X) \rvert$. In this case,
\begin{align} \label{eq:pred-abs-res}
    \hat{C}(X^{n+1}) = \bigtimes_{j=1}^d \left[ \hat{f}_j(X) - W_j, \hat{f}_j(X) + W_j \right].
\end{align}
Alternatively, for quantile-based residuals $E_j(X,Y) = |\max \bigl\{\hat{f}_j^{\mathrm{lo}}(X) - Y_j,\; Y_j - \hat{f}_j^{\mathrm{hi}}(X)\bigr\}|$,
\begin{equation}
    \hat{C}(X^{n+1}) = \bigtimes_{j=1}^d \left[ \hat{f}_j^{\mathrm{lo}}(X^{n+1}) - W_j, \hat{f}_j^{\mathrm{hi}}(X^{n+1}) + W_j \right].
    \label{eq:pred-abs-qr}
\end{equation}
In the following, we will describe practical approximations of these oracle prediction sets, starting from intuitive but not fully satisfactory solutions before presenting our method.

\subsection{Simple but Unsatisfactory Practical Approaches} \label{sec:data-splitting}

An intuitive {\em plug-in} approximation of the ideal oracle is obtained by replacing the population mean $\mu^*$ and standard deviation $\sigma^*$ parameters in~\eqref{eq:oracle-mean-std} with their empirical counterparts, $\hat{\mu}^{\mathrm{pi}}$ and $\hat{\sigma}^{\mathrm{pi}}$, evaluated on the calibration data, and then constructing hyper-rectangular prediction sets $\hat{C}^{\mathrm{pi}}(X^{n+1})$ in the form of~\eqref{eq:joint-prediction} using $W_j = \hat{Q}^{\mathrm{pi}}_{1-\alpha} \cdot \hat{\sigma}^{\mathrm{pi}}_j + \hat{\mu}^{\mathrm{pi}}_j$, where $\hat{Q}^{\mathrm{pi}}_{1-\alpha} := \hat{Q}_{1-\alpha}(S^1_{\mathrm{pi}},\ldots,S^n_{\mathrm{pi}})$ and $S^i_{\mathrm{pi}}:= \Phi(E^i; \hat{\mu}^{\mathrm{pi}}, \hat{\sigma}^{\mathrm{pi}})$ for all $i \in [n]$. This approach, outlined by Algorithm~\ref{alg:std-plug-in} in Appendix~\ref{app:algorithms}, tends to perform well in large-samples (as shown empirically in Appendix~\ref{app:additional-experiments}) but lacks finite-sample guarantees.
In fact, using a data-dependent transformation $\Phi$, which breaks the exchangeability between calibration and test scores in $(S^1,\ldots,S^{n+1})$, violates the key assumption of Theorem~\ref{thm:fixed-transformation-coverage}.

An alternative solution that is also intuitive yet unsatisfactory consists of combining the above plug-in approach with an additional data splitting step that allows one to recover finite-sample guarantees: the location and scale parameter estimates for the transformation function $\Phi$, denoted here as $\hat{\mu}^{\mathrm{ds}}$ and $\hat{\sigma}^{\mathrm{ds}}$, are estimated empirically using only a random subset of $n_1 < n$ observations. The remaining $n_2 = n-n_1$ data points are used to evaluate the scores $(S^1_{\mathrm{ds}},\ldots,S^{n_2}_{\mathrm{ds}})$ and compute $\hat{Q}^{\mathrm{ds}}_{1-\alpha} := \hat{Q}_{1-\alpha}(S^1_{\mathrm{ds}},\ldots,S^{n_2}_{\mathrm{ds}})$, finally leading to $W_j = \hat{Q}^{\mathrm{ds}}_{1-\alpha} \cdot \hat{\sigma}^{\mathrm{ds}}_j + \hat{\mu}^{\mathrm{ds}}_j$.
This approach is outlined by Algorithm~\ref{alg:std-ds} in Appendix~\ref{app:algorithms}.
While it satisfies the key assumption of Theorem~\ref{thm:fixed-transformation-coverage}, it makes an inefficient use of the data. Since it reduces the effective number of calibration samples roughly by half, it tends to produce more unstable and often wider-than-necessary prediction sets, especially in small-sample regimes.

The main contribution of this paper is therefore to overcome this inefficiency. We develop a \emph{transductive} conformal prediction method that uses a similar standardization idea but avoids wasteful data splitting by reusing the entire calibration data to estimate the scaling parameters, all while retaining finite-sample guarantees and manageable computational costs.
Achieving this requires resolving several technical challenges, as explained step by step below.

\section{Methodology} \label{sec:methods}

\subsection{A Transductive Rescaling Oracle and High-Level Method Blueprint}\label{sec:method-oracle}

According to Theorem \ref{thm:fixed-transformation-coverage}, prediction sets of the form \eqref{eq:joint-prediction-scalar} achieve finite-sample marginal coverage if the dimension-reduction map $\Phi$ is independent of the calibration data, as assumed in the previous section, or if $(\Phi(E^1),\ldots,\Phi(E^{n+1}))$ can remain exchangeable conditional on $\Phi$.
This latter requirement, which is strictly weaker than independence, is the foundation of the method developed in this section.
We begin by showing how this conditional-exchangeability condition naturally leads to a different \emph{transductive} oracle, which is much more closely aligned with the practical method introduced below.

Consider an imaginary oracle that knows the unordered values in $\{E^1,\ldots,E^{n+1}\}$; she would estimate the population location and scale parameters $\mu^*, \sigma^*$ defined in~\eqref{eq:oracle-mean-std} as follows:
\begin{align} \label{eq:transductive-oracle-mean-std}
& \hat{\mu}^{\mathrm{oracle}}_j  = \frac{1}{n+1} \sum_{i=1}^{n+1} E^i_j,
& \hat{\sigma}^{\mathrm{oracle}}_j 
= \sqrt{\frac{1}{n} \sum_{i=1}^{n+1} \!\bigl(E^i_j - \hat{\mu}^{\mathrm{oracle}}_j\bigr)^2}, 
&& \forall j \in [d].
\end{align}
Since these statistics are symmetric in all $n+1$ data points and thus invariant to their ordering, conditioning on the random function $\Phi(\cdot; \hat{\mu}^{\mathrm{oracle}}, \hat{\sigma}^{\mathrm{oracle}} )$ maintains the exchangeability of $(E^1,\ldots,E^{n+1})$.
Therefore, Theorem~\ref{thm:fixed-transformation-coverage} implies the following hyper-rectangular prediction sets guarantee finite-sample marginal coverage: $\hat{C}^{\mathrm{oracle}}(X^{n+1})$ in the form of~\eqref{eq:joint-prediction} using $W_j = \hat{Q}^{\mathrm{oracle}}_{1-\alpha} \cdot \hat{\sigma}^{\mathrm{oracle}}_j + \hat{\mu}^{\mathrm{oracle}}_j$, where $\hat{Q}^{\mathrm{oracle}}_{1-\alpha} := \hat{Q}_{1-\alpha}(S^1_{\mathrm{oracle}},\ldots,S^n_{\mathrm{oracle}})$ and $S^i_{\mathrm{oracle}}:= \Phi(E^i; \hat{\mu}^{\mathrm{oracle}}, \hat{\sigma}^{\mathrm{oracle}})$ for all $i \in [n]$. 
Of course, this oracle is impractical because both $\Phi(\cdot; \hat{\mu}^{\mathrm{oracle}}, \hat{\sigma}^{\mathrm{oracle}})$ and $\hat{Q}^{\mathrm{oracle}}_{1-\alpha}$ depend on the unknown test outcome $Y^{n+1}$ through $E^{n+1}$. 

We will now translate this oracle into a practical method producing a (slightly) larger prediction set that depends only on the observable residuals $E^1,\ldots,E^{n}$, through the statistics
\begin{align} \label{eq:plug-in-mean-std}
& \hat{\mu}_j = \frac{1}{n} \sum_{i=1}^{n} E^i_j,
& \hat{\sigma}_j =  \sqrt{\frac{1}{n} \sum_{i=1}^{n} \!\bigl(E^i_j - \hat{\mu}_j\bigr)^2},
&& \forall j \in [d].
\end{align}
Note that the choice of using a factor $1/n$ instead of the common $1/(n-1)$ in the above definition of $\hat{\sigma}$ is deliberate, as it simplifies the derivation of the following results. Note also that $0 < \hat{\sigma}_j < \infty$ almost surely under Assumption~\ref{eq:assumption-scores}, as long as $n > 1$.

To obtain a practical method, we proceed in two steps, separately addressing the unknown nature of $\Phi(\cdot; \hat{\mu}^{\mathrm{oracle}}, \hat{\sigma}^{\mathrm{oracle}})$ and $\hat{Q}^{\mathrm{oracle}}_{1-\alpha}$. First, we will find suitable {\em monotonically nondecreasing} {\em link functions} $\omega_1(\cdot), \ldots, \omega_d(\cdot) : \mathbb{R} \mapsto \mathbb{R}_+$ satisfying 
\begin{equation}
S_{\mathrm{oracle}}^{n+1} \le c 
\iff 
0\leq E^{n+1}_j \le \omega_j(c), \quad \forall c \in \mathbb{R}.
\label{eq:key-ineq}
\end{equation}
Substituting $c = \hat{Q}^{\mathrm{oracle}}_{1-\alpha}$ into~\eqref{eq:key-ineq} gives the equivalence relation
\begin{equation} \label{eq:oracle-dual}
Y^{n+1} \in \hat{C}^{\mathrm{oracle}}(X^{n+1}) \iff Y^{n+1} \in \tilde{C}(X^{n+1}),
\end{equation}
where the set $\tilde{C}(X^{n+1})$ is defined as
\begin{equation}
\tilde{C}(X^{n+1})
:= 
\left\{
y \in \mathbb{R}^d :
0\leq E_j(X^{n+1}, y) \leq \omega_j(\hat{Q}^{\mathrm{oracle}}_{1-\alpha}), \, \forall j \in [d]
\right\}.
\label{eq:oracle-prediction}
\end{equation}
Although $\tilde{C}(X^{n+1})$ is still impractical, it is closer to being computable.
The last remaining step is to replace $\hat{Q}^{\mathrm{oracle}}_{1-\alpha}$ with a conservative data-driven estimate, which we denote here as $\tilde{Q}_{1-\alpha} \geq \hat{Q}^{\mathrm{oracle}}_{1-\alpha}$. Doing that will give a practical prediction set
\begin{equation} \label{eq:trans-oracle-upper-bound}
\hat{C}(X^{n+1})
:= 
\left\{
y \in \mathbb{R}^d :
0\leq E_j(X^{n+1}, y) \leq \omega_j(\tilde{Q}_{1-\alpha}), \, \forall j \in [d]
\right\}
\supseteq \tilde{C}(X^{n+1}),
\end{equation}
inheriting the coverage lower bound from Theorem~\ref{thm:fixed-transformation-coverage}.
Technical difficulties, which we address below, arise when one tries to keep this approximation tight and computationally efficient.

\subsection{Explicit Form of the Link Functions}

The following result gives an explicit form for the vector of link functions $\omega := (\omega_1(\cdot), \ldots, \omega_d(\cdot))$ defined implicitly in~\eqref{eq:key-ineq}.

\begin{lemma}
\label{lem:solution-key-ineq}
Under Assumption~\ref{eq:assumption-scores}, suitable link functions $\omega_1(c), \ldots ,\omega_d(c)$ satisfying~\eqref{eq:key-ineq} are:
\[
    \omega_j(c)
    =
    \begin{cases}
    0, & \text{if } c \leq -\frac{n}{\sqrt{n+1}},\\
    \max\!\left\{0,\,
    \hat{\mu}_j - \hat{\sigma}_j |c|
    \sqrt{\dfrac{(n+1)^2}{n^2 - (n+1)c^2}}\right\},
    & \text{if } -\dfrac{n}{\sqrt{n+1}} < c < 0, \\[1.1em]
    \hat{\mu}_j + \hat{\sigma}_j |c|
    \sqrt{\dfrac{(n+1)^2}{n^2 - (n+1)c^2}},
    & \text{if } 0 \le c < \dfrac{n}{\sqrt{n+1}}, \\[1.1em]
    \infty, & \text{if } c \ge \dfrac{n}{\sqrt{n+1}},
    \end{cases}\quad\forall j \in [d].
\]
\end{lemma}

Replacing this expression for the link functions into~\eqref{eq:oracle-prediction} gives an upper bound $\omega_j(\hat{Q}^{\mathrm{oracle}}_{1-\alpha})$ for $E_j(X^{n+1}, y)$ that is finite (and hence informative) if and only if $\hat{Q}^{\mathrm{oracle}}_{1-\alpha} <n/\sqrt{n+1}$. Fortunately, $\hat{Q}^{\mathrm{oracle}}_{1-\alpha} <n/\sqrt{n+1}$ almost surely as long as the sample size is not too small.

\begin{lemma}
Under Assumption~\ref{eq:assumption-scores}, if $n \geq 1/\alpha -1 $, then $\hat{Q}^{\mathrm{oracle}}_{1-\alpha} < n/\sqrt{n+1}$ almost surely. 
    \label{lem:finite-prediction-guarantee}
\end{lemma}

Armed with Lemma~\ref{lem:solution-key-ineq} and Lemma~\ref{lem:finite-prediction-guarantee}, what remains to be done is to obtain a practical prediction set $\hat{C}(X^{n+1}) \supseteq \hat{C}^{\mathrm{oracle}}(X^{n+1})$, as previewed in~\eqref{eq:trans-oracle-upper-bound}, is to find a tight upper bound $\tilde{Q}_{1-\alpha}$ for $\hat{Q}^{\mathrm{oracle}}_{1-\alpha}$. This is the main challenge and focus of the next sections.

\subsection{A Conservative Global-Worst-Case Approach}
\label{sec:method-gwc}

We introduce a (highly) conservative approximation of $\tilde{C}(X^{n+1})$ from \eqref{eq:oracle-prediction}, obtained by taking the worst-case upper bound of $\hat{Q}^{\mathrm{oracle}}_{1-\alpha}$ over all possible values of $E^{n+1} \in \mathbb{R}^d_+$.
We refer to the resulting {\em global-worst-case} (GWC) prediction set as $\hat{C}^{\mathrm{gwc}}(X^{n+1})$.
Because $\hat{C}^{\mathrm{gwc}}(X^{n+1}) \supseteq \tilde{C}(X^{n+1})$ almost surely, Theorem \ref{thm:fixed-transformation-coverage} ensures that $\hat{C}^{\mathrm{gwc}}(X^{n+1})$ attains marginal coverage of at least $1-\alpha$ in finite samples.
As we show in the next section, this conservative prediction set can in fact be sharpened considerably, without sacrificing marginal coverage, albeit with some additional methodological effort.

Let $z \in \mathbb{R}^d_+$ denote a placeholder for the test residual $E^{n+1}$. 
Define the plug-in scaling parameters $\hat{\mu}(z)$ and $\hat{\sigma}(z)$ that would result if $E^{n+1}$ in~\eqref{eq:transductive-oracle-mean-std} were replaced by $z$:
\begin{align} \label{eq:placeholder-mean-std}
\begin{split}
\hat{\mu}_j(z_j) = \frac{\sum_{i=1}^{n} E^i_j + z_j}{n+1}, \quad
\hat{\sigma}_j(z_j) = \sqrt{\frac{\sum_{i=1}^{n} \!\bigl(E^i_j - \hat{\mu}_j(z_j)\bigr)^2 + \bigl(z_j - \hat{\mu}_j(z_j)\bigr)^2}{n} }, \quad
 j \in [d].
\end{split}
\end{align}
With this notation, we have $\hat{\mu}_j(E^{n+1}_j) = \hat{\mu}^{\mathrm{oracle}}_j$ and $\hat{\sigma}_j(E^{n+1}_j) = \hat{\sigma}^{\mathrm{oracle}}_j$.
Define then the GWC transformation $\hat{\Phi}^{\mathrm{gwc}}: \R^d_+ \mapsto \R$ as
\begin{equation}
\hat{\Phi}^{\mathrm{gwc}}(t)
:= 
\max_{1 \le j \le d} 
\sup_{z_j \ge 0} 
\frac{t_j - \hat{\mu}_j(z_j)}{\hat{\sigma}_j(z_j)} \geq \Phi(t; \hat{\mu}^{\mathrm{oracle}}, \hat{\sigma}^{\mathrm{oracle}})
\qquad \forall t \in \mathbb{R}^d_+.
\label{eq:gwc-residual-transformation}
\end{equation}
Intuitively, for each coordinate $j \in [d]$ and $t \in \R^d_+$, we consider the most unfavorable (largest) possible value of the standardized ratio as a function of $E^{n+1}_j = z_j\geq 0$. 
This ensures that $\hat{\Phi}^{\mathrm{gwc}}(t)$ dominates the oracle transformation $\Phi(t; \hat{\mu}^{\mathrm{oracle}}, \hat{\sigma}^{\mathrm{oracle}})$ uniformly over all $t$.
Conveniently, $\hat{\Phi}^{\mathrm{gwc}}(t)$ admits a closed-form expression.

\begin{lemma}
\label{lem:gwc-rescaling}
The GWC transformation $\hat{\Phi}^{\mathrm{gwc}}$ defined in~\eqref{eq:gwc-residual-transformation} can be written as:
\[
\hat{\Phi}^{\mathrm{gwc}}(t)
= 
\max_{1 \le j \le d} 
\max\!\left\{
\frac{t_j - \hat{\mu}_j(0)}{\hat{\sigma}_j(0)},\,
\frac{t_j - \hat{\mu}_j(z_j^*)}{\hat{\sigma}_j(z_j^*)},\,
-\frac{1}{\sqrt{n+1}}
\right\}.
\]
where $z^*_j = \hat\mu_j - \frac{\hat\sigma_j^2}{t_j-\hat\mu_j}$ if $\hat\mu_j \geq \frac{\hat\sigma_j^2}{t_j-\hat\mu_j}$ and $z^*_j = 0$ otherwise.
\end{lemma}

Using Lemma~\ref{lem:gwc-rescaling}, we compute the scalar non-conformity scores $S^i_{\mathrm{gwc}} := \hat{\Phi}^{\mathrm{gwc}}(E^i)$ for $i \in [n]$, and define the corresponding conformal quantile  $\hat{Q}^{\mathrm{gwc}}_{1-\alpha}: = \hat{Q}_{1-\alpha}(S^1_{\mathrm{gwc}},\ldots, S^n_{\mathrm{gwc}})$.
Because $\hat{\Phi}^{\mathrm{gwc}}(t) \ge \Phi(t; \hat{\mu}^{\mathrm{oracle}}, \hat{\sigma}^{\mathrm{oracle}})$ for all $t$, it follows that 
$\hat{Q}^{\mathrm{gwc}}_{1-\alpha} \geq \hat{Q}^{\mathrm{oracle}}_{1-\alpha}$ almost surely. Then the GWC prediction set is given by:
\begin{equation}
\hat{C}^{\mathrm{gwc}}(X^{n+1})
:= 
\left\{y \in \mathbb{R}^d : 0\leq E(X^{n+1},y) \leq \omega(\hat{Q}^{\mathrm{gwc}}_{1-\alpha}), \quad \forall j \in [d]\right\}.
\label{eq:gwc-prediction}
\end{equation}
Algorithm~\ref{alg:global} in Appendix~\ref{app:algorithms} summarizes this procedure, which guarantees valid coverage.
\begin{theorem} \label{thm:TSCP-GWC-coverage}
Assume $(E^1,\ldots,E^{n+1})$ are exchangeable and Assumption~\ref{eq:assumption-scores} holds.
Then, $\hat{C}^{\mathrm{gwc}}(X^{n+1})$ constructed by Algorithm~\ref{alg:global} satisfies: $\mathbb{P}[ Y^{n+1} \in \hat{C}^{\mathrm{gwc}}(X^{n+1}) ] \ge 1 - \alpha$.
\end{theorem}
\begin{proof}
From~\eqref{eq:gwc-prediction}, it is easy to see that $\hat{C}^{\mathrm{gwc}}(X^{n+1}) \supseteq \tilde{C}(X^{n+1})$ almost surely because $\hat{Q}^{\mathrm{gwc}}_{1-\alpha} \geq \hat{Q}^{\mathrm{oracle}}_{1-\alpha}$ and the link function $\omega$ outlined in Lemma~\ref{lem:solution-key-ineq} is monotone non-decreasing. We already know from Theorem \ref{thm:fixed-transformation-coverage} then $\mathbb{P}[ Y^{n+1} \in \tilde{C}(X^{n+1}) ] \ge 1 - \alpha$.
\end{proof}

We refer to Appendix~\ref{app:gwc-cost} for details on how to construct $\hat{C}^{\mathrm{gwc}}(X^{n+1})$ in practice at cost $\mathcal{O}(dn)$. While computationally efficient, however, $\hat{C}^{\mathrm{gwc}}(X^{n+1})$ is {\em statistically} inefficient as it tends to be substantially larger than the estimated oracle prediction set $\tilde{C}(X^{n+1})$, which we would ideally like to approximate tightly.
In the next section, we therefore extend this approach and develop a tighter {\em local worst-case} (LWC) prediction set $\hat{C}^{\mathrm{lwc}}(X^{n+1}) \subseteq \hat{C}^\mathrm{gwc}(X^{n+1})$ that approximates the oracle $\tilde{C}(X^{n+1})$ more closely. 

\subsection{A Tighter Local-Worst-Case Method}
\label{sec:method-lwc}

The construction of $\hat{C}^{\mathrm{gwc}}(X^{n+1})$ described in the previous section relies on a {\em global} upper bound $\hat{Q}^{\mathrm{gwc}}_{1-\alpha}$ for the ideal quantile $\hat{Q}^{\mathrm{oracle}}_{1-\alpha}$, designed to hold uniformly over all possible values of the unobserved test residual $E^{n+1} \in \mathbb{R}^d_+$. However, the resulting prediction sets are often overly conservative, as the GWC procedure ignores the fact that many choices of $E^{n+1}$ correspond to outcome values $Y^{n+1} \in \mathbb{R}^d$ that are extremely unlikely to belong to the ideal prediction set $\tilde{C}(X^{n+1})$ and could therefore be safely disregarded.

This motivates the development of a more statistically efficient approach that constructs and aggregates a collection of {\em local} prediction sets, each tailored to a distinct working hypothesis that restricts $Y^{n+1}$ to a specific region of the outcome space. After a high-level overview of the main ideas and an implicit characterization of the target prediction set for a given partition of the outcome space into disjoint local regions, we detail the implementation of the method in three steps. First, we introduce an intuitive, data-driven partition guided by the order statistics of the calibration residuals. Second, we derive a tractable explicit characterization of the resulting local prediction sets. Third, we develop a computationally efficient algorithm for constructing a tight enclosing rectangular prediction set in practice, while avoiding explicit enumeration of an exponential number of local prediction sets.

\subsubsection{Method Overview}

For any subset $\mathcal{R} \subseteq \R^d$ of the outcome space, consider the working hypothesis that $Y^{n+1} \in \mathcal{R}$.
Let $\hat{Q}^\mathrm{lwc}_{1-\alpha}(\mathcal{R})$ denote an upper bound for $\hat{Q}^{\mathrm{oracle}}_{1-\alpha}$ tailored to this hypothesis, satisfying
\begin{align*}
  Y^{n+1} \in \mathcal{R} \quad \Longrightarrow \quad \hat{Q}^\mathrm{lwc}_{1-\alpha}(\mathcal{R}) \geq \hat{Q}^{\mathrm{oracle}}_{1-\alpha}.
\end{align*}
Then, define the $\mathcal{R}$-specific local prediction set $\hat{C}^\mathrm{lwc} (X^{n+1}; \mathcal{R})$ by replacing the (very conservative) upper bound $\hat{Q}^{\mathrm{gwc}}_{1-\alpha}$ in~\eqref{eq:gwc-prediction} with its local alternative $\hat{Q}^\mathrm{lwc}_{1-\alpha}(\mathcal{R})$:
\begin{align} \label{eq:lwc-prediction}
\hat{C}^\mathrm{lwc} (X^{n+1}; \mathcal{R}) := \left\{y\in \mathcal{R}: 0\leq E_j(X^{n+1},y) \leq \omega_j(\hat{Q}^\mathrm{lwc}_{1-\alpha}(\mathcal
{R})) \right\}.
\end{align}
By construction, this satisfies
\begin{align} \label{eq:lwc-local-containment}
  Y^{n+1} \in \mathcal{R} \quad \Longrightarrow \quad \mathcal{R} \supseteq \hat{C}^\mathrm{lwc} (X^{n+1}; \mathcal{R}) \supseteq \tilde{C}(X^{n+1}) \cap \mathcal{R}.
\end{align}

Although we cannot know \emph{a priori} whether $Y^{n+1} \in \mathcal{R}$ for any fixed region $\mathcal{R}$, we can construct many local prediction sets as in~\eqref{eq:lwc-prediction} using a collection $\mathcal{P}$ of hypothesized regions $\mathcal{R}$ chosen to be sufficiently rich so that $Y^{n+1} \in \mathcal{R}$ with high probability for some $\mathcal{R} \in \mathcal{P}$, and then aggregate these local prediction sets.
To ensure interpretability, we report a relatively tight rectangular set that bounds the union of all localized prediction sets, rather than the union itself.
Moreover, since the conservative prediction set $\hat{C}^{\mathrm{gwc}}(X^{n+1})$ from the previous section already enjoys valid marginal coverage, our goal is to obtain a smaller prediction set. Accordingly, it suffices to consider a {\em random partition} $\mathcal{P}$ of $\hat{C}^{\mathrm{gwc}}(X^{n+1})$, rather than of the entire outcome space $\mathbb{R}^d$. 
In summary, we construct $\hat{C}^{\mathrm{lwc}}(X^{n+1})$ using:
\begin{equation}
    \begin{aligned}
        \hat{C}^{\mathrm{lwc}}(X^{n+1}) &:= \mathrm{Rect}\left(\bigsqcup_{\mathcal{R} \in \mathcal{P}} \hat{C}^\mathrm{lwc}(X^{n+1}; \mathcal{R}); \mathcal{P}\right),
    \label{eq:lwc-joint-union}
    \end{aligned}
\end{equation}
where $\mathrm{Rect}(A; \mathcal{P})$ denotes a rectangle containing the subset $A \subseteq \mathbb{R}^d$ that does not extend beyond the outcome space partitioned by $\mathcal{P}$ in any coordinate. 
In practice, our method constructs a {\em relatively tight}, though not necessarily the {\em tightest}, such rectangle due to computational considerations. Nonetheless, we will show that the resulting prediction sets are (often much) more informative than  $\hat{C}^{\mathrm{gwc}}(X^{n+1})$ in finite samples. Moreover, in the large-sample limit, the union in~\eqref{eq:lwc-joint-union} becomes approximately rectangular itself, making this enclosing step a very mild relaxation from the perspective of statistical efficiency.

The following result guarantees that this high-level approach leads to prediction sets $\hat{C}^{\mathrm{lwc}}(X^{n+1})$ that are contained in $\hat{C}^{\mathrm{gwc}}(X^{n+1})$ while retaining valid marginal coverage.
\begin{proposition}
Consider a (random) collection $\mathcal{P}$ of disjoint regions $\mathcal{R} \subseteq \mathbb{R}^d$ satisfying $\sqcup_{\mathcal{R} \in \mathcal{P} } \mathcal{R} = \hat{C}^{\mathrm{gwc}}(X^{n+1})$ almost surely.
For any $\mathcal{R} \in \mathcal{P}$, let $\hat{C}^\mathrm{lwc}(X^{n+1}; \mathcal{R}) \subseteq \mathbb{R}^d$ denote a local prediction set satisfying~\eqref{eq:lwc-local-containment} almost surely.
Then, $\hat{C}^{\mathrm{lwc}}(X^{n+1})$ defined in~\eqref{eq:lwc-joint-union}  satisfies $\hat{C}^{\mathrm{lwc}}(X^{n+1}) \subseteq \hat{C}^{\mathrm{gwc}}(X^{n+1})$.
Therefore, if $(E^1,\ldots,E^{n+1})$ are exchangeable and Assumption~\ref{eq:assumption-scores} holds, $\mathbb{P}[ Y^{n+1} \in \hat{C}^{\mathrm{lwc}}(X^{n+1}) ] \ge 1 - \alpha$.
\label{prop:lwc-approximation}
\end{proposition}

We next introduce an intuitive data-driven partition based on the possible rankings of the test residuals $E_j^{n+1}$ relative to the calibration residuals $(E_j^{1},\ldots,E_j^{n})$ for each coordinate $j \in [d]$, and then we characterize the corresponding prediction set $\hat{C}^{\mathrm{lwc}}(X^{n+1})$.

\subsubsection{Step 1: A Data-Driven Partition Based on Sample Order Statistics}

For each $j \in [d]$, let $E^{(k)}_j$ denote the $k$-th order statistic of $\{E^i_j\}_{i=1}^n$, so that $0 := E^{(0)}_j < E^{(1)}_j \leq \ldots \leq E^{(n)}_j < E^{(n+1)}_j := +\infty$. 
For each multi-index $h=(h_1,\ldots,h_d)\in [n+1]^d$, define a (half-open) rectangle $R^h$ in the residual space and its region $\mathcal{Y}^h$ in the outcome space as:
\begin{equation}
    \label{eq:rectangle-explicit}
    \begin{aligned}
        R^h 
        &:=  \bigtimes_{j=1}^d \bigl[\,E^{(h_j-1)}_j,\, U^{h_j}_j\,\bigr), \quad\text{where} \quad U^{h_j}_j := \min\left\{E^{(h_j)}_j,  \omega_j(\hat{Q}^{\mathrm{gwc}})\right\},\\
        \mathcal{Y}^h &:= \left\{y \in \R^d: E(X^{n+1}, y) \in R^h\right\}.
    \end{aligned}
\end{equation}
This defines a partition $ \mathcal{P} = \{\mathcal{Y}^h\}_{h\in[n+1]^d}$ of $\hat{C}^\mathrm{gwc}(X^{n+1})$ into disjoint regions, some of which may be empty, depending on the values of $\{\omega_j(\hat{Q}^{\mathrm{gwc}})\}^d_{j=1}$.
Figure~\ref{fig:global-to-local} visualizes the partition and the corresponding construction of $\hat{C}^{\mathrm{lwc}}(X^{n+1})$, in a two-dimensional toy example.

\begin{figure}[!htb]
  \centering
  \includegraphics[width = 0.7\textwidth]{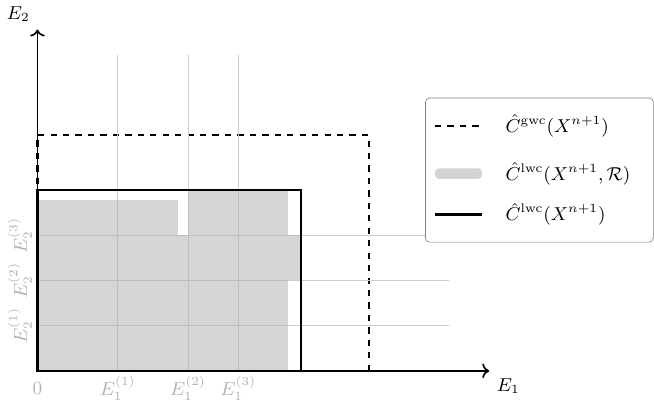}
\caption{Illustration of the construction of transductively standardized conformal prediction sets using the method outlined in~\eqref{eq:lwc-joint-union}, with the data-driven partition $\mathcal{P}$ of the outcome space defined in~\eqref{eq:rectangle-explicit}, in a toy example with a two-dimensional outcome variable.
The gray grid lines indicate the partition $\mathcal{P}$ in the residual space, as determined by the order statistics $E^{(k)}_j$ of the calibration residuals.
The shaded rectangles represent local prediction sets $\hat{C}^\mathrm{lwc}(X^{n+1}; \mathcal{R})$ corresponding to different hypothesized rectangles $\mathcal{R} \in \mathcal{P}$.
The solid rectangle shows the projection, in residual space, of their rectangular enclosure defined in~\eqref{eq:lwc-joint-union}, which constitutes our final prediction set.
For comparison, the larger dashed rectangle shows the corresponding projection of the more conservative prediction set $\hat{C}^\mathrm{gwc}(X^{n+1})$ from Section~\ref{sec:method-gwc}.}
  \label{fig:global-to-local}
\end{figure}

While it may appear at first sight that this choice of partition necessarily incurs exponential computational costs of order $\mathcal{O}(n^d)$, rendering it intractable for all but very small dimensions $d$, in fact it enables a much more efficient characterization of the local prediction sets $\hat{C}^{\mathrm{lwc}}(X^{n+1}; h)$ that will lead to a practical algorithm for constructing $\hat{C}^{\mathrm{lwc}}(X^{n+1})$ at a computational cost of only $\mathcal{O}(d^2 n \log n)$.

\subsubsection{Step 2: Characterization of the Local Prediction Sets}

For any $h \in [n+1]^d$ with $\mathcal{Y}^h \neq \emptyset$, we characterize the local prediction set $\hat{C}^{\mathrm{lwc}}(X^{n+1};h)$ as 
\begin{equation}
\label{eq:lwc-prediction-h}
\hat{C}^{\mathrm{lwc}}(X^{n+1};h)
:= \left\{ y \in \mathcal{Y}^h : 0\leq E_j(X^{n+1},y) \leq  \omega_j(\hat{Q}^{\mathrm{lwc}}_{1-\alpha}(h)) \right\},
\end{equation}
where the function $\omega_j$ is given in Lemma~\ref{lem:solution-key-ineq}, and the scalar quantity $\hat{Q}^{\mathrm{lwc}}_{1-\alpha}(h)$ is computed as described below, following a strategy similar to the worst-case approach from Section~\ref{sec:method-gwc}, but operating more liberally under the working hypothesis that $Y^{n+1} \in \mathcal{Y}^h$.

Define the LWC transformation function $\hat{\Phi}^{\mathrm{lwc}}(\cdot, h): \R^d_+ \mapsto \R$ such that, for any $t \in \mathbb{R}^d_+$,
\begin{align}
\label{eq:lwc-residual-transformation}
& \hat{\Phi}^{\mathrm{lwc}}(t;h)
:= \max_{1 \le j \le d}
\left\{
\sup_{z \in R^h}\frac{t_j}{\hat{\sigma}_j(z_j)}
\;-\;
\inf_{z \in \mathcal{E}^{\mathrm{gwc}}}\frac{\hat{\mu}_j(z_j)}{\hat{\sigma}_j(z_j)}
\right\},
& \mathcal{E}^{\mathrm{gwc}}
:= 
\bigtimes_{j=1}^{d} \left[0,\, \omega(\hat{Q}^{\mathrm{gwc}}_{1-\alpha})\right]
\end{align}
where $\mathcal{E}^{\mathrm{gwc}}$ is the projection of $\hat{C}^\mathrm{gwc}(X^{n+1})$ onto the residual space; i.e., $y \in \hat{C}^\mathrm{gwc}(X^{n+1}) \iff E(X^{n+1},y) \in \mathcal{E}^{\mathrm{gwc}}$. 
With this choice, note that, whenever $E^{n+1}\in R^h$,
\begin{equation}
  \hat{\Phi}^{\mathrm{lwc}}(t;h)
\;\ge\;
\max_{1 \le j \le d}\;
\sup_{z\in R^h}\frac{t_j-\hat{\mu}_j(z_j)}{\hat{\sigma}_j(z_j)} \geq \Phi\!\bigl(t;\hat{\mu}^{\mathrm{oracle}},\hat{\sigma}^{\mathrm{oracle}}\bigr),
\quad \forall\, t\in\mathbb{R}^d_+.
\label{eq:lwc-transformation-bound}
\end{equation}
Therefore, $\hat{\Phi}^{\mathrm{lwc}}(t;h)$ is {\em locally} more conservative than the oracle transformation from Section~\ref{sec:method-oracle} under the working hypothesis that $Y^{n+1} \in \mathcal{Y}^h$, while being much less conservative than $\hat{\Phi}^{\mathrm{gwc}}(t)$ from Section~\ref{sec:method-gwc}.
Moreover, $\hat{\Phi}^{\mathrm{lwc}}(t;h)$ admits the following closed form.
\begin{lemma}
\label{lem:lwc-rescaling}
If $\mathcal{Y}^h\neq \emptyset$, then $\hat{\Phi}^{\mathrm{lwc}}(t; h)$ defined in~\eqref{eq:lwc-residual-transformation} can be written for any $t\in\mathbb{R}^d_+$ as:
\[
\hat{\Phi}^{\mathrm{lwc}}(t;h)
= \max_{1 \le j \le d}
\left\{
\frac{t_j}{r_j(h_j)}
-
\min\!\left\{
\frac{\hat{\mu}_j(0)}{\hat{\sigma}_j(0)},\;
\lim_{z_j \uparrow \omega_j(\hat{Q}^{\mathrm{gwc}}_{1-\alpha})}
\frac{\hat{\mu}_j(z_j)}{\hat{\sigma}_j(z_j)}
\right\}
\right\},
\]
where
\begin{align*}
r_j(h_j) = \begin{cases}
  \hat{\sigma}_j, & \text{if } E^{(h_j-1)}_j \leq \hat{\mu}_j < U^{h_j}_j, \\
  \min\left\{\hat{\sigma}_j(E^{(h_j-1)}_j), \hat{\sigma}_j(U^{h_j}_j)\right\}, & \text{otherwise.}
\end{cases}
\end{align*}
\end{lemma}
Note that the expression for $\hat{\Phi}^{\mathrm{lwc}}(t;h)$ in Lemma~\ref{lem:lwc-rescaling} involves a limit for $z_j \uparrow \omega_j(\hat{Q}^{\mathrm{gwc}}_{1-\alpha})$ because $\omega_j(\hat{Q}^{\mathrm{gwc}}_{1-\alpha})$ can be infinite when $\hat{Q}^{\mathrm{gwc}}_{1-\alpha} \geq \frac{n}{\sqrt{n+1}}$ (recall Lemma~\ref{lem:solution-key-ineq}), in which case $\mathcal{E}^\mathrm{gwc} = \R^d_+$. 
Nonetheless, our construction remains well-defined even in that special case.

The formula for $\hat{\Phi}^{\mathrm{lwc}}(t;h)$ given by Lemma~\ref{lem:lwc-rescaling} makes it easy to compute the scores $S^i_{\mathrm{lwc}}(h) := \hat{\Phi}^{\mathrm{lwc}}(E^i,h)$ for all $i\in [n]$, and to evaluate the empirical quantile  $\hat{Q}^{\mathrm{lwc}}_{1-\alpha}(h):= \hat{Q}_{1-\alpha}(S^1_{\mathrm{lwc}}(h), \ldots, S^n_{\mathrm{lwc}}(h))$ used in the characterization of $\hat{C}^{\mathrm{lwc}}(X^{n+1};h)$ in~\eqref{eq:lwc-prediction-h}.
Under the working hypothesis $Y^{n+1} \in \mathcal{Y}^h$ and $E^{n+1} \in R^h$, it follows immediately from~\eqref{eq:lwc-transformation-bound} that 
$\hat{Q}^{\mathrm{lwc}}_{1-\alpha}(h)\ge \hat{Q}^{\mathrm{oracle}}_{1-\alpha}$ almost surely.
Therefore, this approach guarantees that  $\hat{C}^{\mathrm{lwc}}(X^{n+1};h)$  satisfies~\eqref{eq:lwc-local-containment} for all $h \in [n+1]^d$---this continues to hold even if $\mathcal{Y}^h = \emptyset$, because in that special case $\hat{C}^{\mathrm{lwc}}(X^{n+1};h) = \tilde{C}(X^{n+1})\cap\mathcal{Y}^h = \mathcal{Y}^h=\emptyset$. 
Algorithm~\ref{alg:local-bound} summarizes this construction. 

\begin{algorithm}[!htbp] 
\caption{Construction of local prediction sets under working hypothesis $Y^{n+1} \in \mathcal{Y}^h$} \label{alg:local-bound} 
\textbf{Input:} Calibration residuals $\{E^i\}_{i=1}^{n}$; target miscoverage level $\alpha\in(0,1)$; global bounds $\omega_1(\hat{Q}^{\mathrm{gwc}}_{1-\alpha}),\ldots,\omega_d(\hat{Q}^{\mathrm{gwc}}_{1-\alpha})$; calibration mean $\hat{\mu}$ and standard deviation\ $\hat{\sigma}$; index $h\in [n+1]^d$.
\begin{algorithmic}[1] 
\If{$\mathcal{Y}^h \neq \emptyset$} 
    \State Compute $S_{\mathrm{lwc}}^i(h) \gets \hat{\Phi}^{\mathrm{lwc}}(E^i,h)$ \textbf{for} $i \in [n]$ using $\hat{\Phi}^{\mathrm{lwc}} (\cdot,h)$ defined in~\eqref{eq:lwc-residual-transformation}.
    \State Compute $\hat{Q}^{\mathrm{lwc}}_{1-\alpha}(h) \gets \hat{Q}_{1-\alpha}(S_{\mathrm{lwc}}^1(h),\ldots,S_{\mathrm{lwc}}^n(h))$. 
    \State \Return $\omega_1(\hat{Q}^{\mathrm{lwc}}_{1-\alpha}(h)),\ldots,\omega_d(\hat{Q}^{\mathrm{lwc}}_{1-\alpha}(h))$ using $\omega$ defined in~\eqref{eq:key-ineq}.
\Else
\State \Return $\emptyset$.
\EndIf 
\end{algorithmic} 
\end{algorithm}

\subsubsection{Computational Shortcut for Rectangular Enclosure}

While the characterization provided above makes it easy to compute $\hat{C}^{\mathrm{lwc}}(X^{n+1};h)$ for any specific index $h \in [n+1]^d$, the problem remains that the expression for $\hat{C}^{\mathrm{lwc}}(X^{n+1})$ in~\eqref{eq:lwc-joint-union} involves an {\em exponential} number of such local prediction sets. Therefore, to obtain a practical method, we must now develop a shortcut that avoids full enumeration of all $\hat{C}^{\mathrm{lwc}}(X^{n+1};h)$.

Our solution begins by noting that $\hat{C}^{\mathrm{lwc}}(X^{n+1};h)$ from~\eqref{eq:lwc-prediction-h} can be re-written as:
\begin{align} \notag
  \hat{C}^\mathrm{lwc}(X^{n+1}; h) = \bigcap_{j=1}^d \left\{ y\in\R^d : E^{(h_j-1)}_j \leq E_j(X^{n+1},y) \leq B^h_j \,\right\} \subseteq \hat{C}^\mathrm{gwc}(X^{n+1}), \\
B^h_j :=
\begin{cases}
\min\left\{U^{h_j}_j,\; \omega_j(\hat{Q}^{\mathrm{lwc}}_{1-\alpha}(h))\right\}, 
& \text{if } U^{h_j}_j,\; \omega_j(\hat{Q}^{\mathrm{lwc}}_{1-\alpha}(h)) > E^{(h_j-1)}_j,\\[0.4em]
0, & \text{otherwise.}
\end{cases}
\quad \forall j \in [d].
\label{eq:rect-wise-bounds}
\end{align}

For every coordinate $j \in [d]$, define the $j$-th boundary
\begin{align} \label{eq:def-Lj}
  \mathcal{L}_j := \max_{\substack{h \in [n+1]^d}}B^h_j.
\end{align}
For example, in the two-dimensional illustration of Figure~\ref{fig:global-to-local}, the value of $\mathcal{L}_1$ represents the most extreme right-side endpoint of all shaded rectangles along the horizontal dimension, while $\mathcal{L}_2$ represents the most extreme upper endpoint of all shaded rectangles along the vertical dimension.
Therefore, $\mathcal{L}_1$ and $\mathcal{L}_2$ uniquely identify the two sides of the solid rectangle shown in Figure~\ref{fig:global-to-local}, which represents the prediction set $\hat{C}^\mathrm{lwc}(X^{n+1})$.
In general, we write
\[
\hat{C}^\mathrm{lwc}(X^{n+1}) := \left\{y \in \R^d: 0\leq E_j(X^{n+1},y) \leq \mathcal{L}_j\right\} 
\supseteq \bigsqcup_{h \in [n+1]^d}\hat{C}^\mathrm{lwc}(X^{n+1};h).
\]
Notably, this rectangular enclosure is typically very tight in practice, because the union of the local prediction sets on the right-hand-side above is itself very close to being a rectangle when the sample size is large enough.

The last remaining step to make our method practical is finding an efficient algorithm for computing the boundaries $\mathcal{L}_j$ for all $j \in [d]$, without requiring enumeration of all local prediction sets. 
Such a shortcut is made possible by the following result, which simplifies the definition of $\mathcal{L}_j $ in~\eqref{eq:def-Lj} from a maximum over $(n+1)^d$ distinct values to a maximum over an explicit, narrowed-down list of only $n+1$ relevant candidates.
Before stating this result, it is helpful to introduce some notation. 

For any $j \in [d]$, define the (random) \emph{mean index} $h_j^* \in [n+1]$ as that unique index such that $E^{(h^*_j-1)}_j \le \hat{\mu}_j \le E^{(h^*_j)}_j$; define also $h^* := (h^*_1, \ldots, h^*_d)$.
Then, for all $j\in[d]$, define 
\[
\mathrm{Row}_j(h^*):=\left\{h\in[n{+}1]^d : h_k = h^*_k \text{ for all } k\neq j \right\}, \quad\forall j \in [d].
\]
In words, $\mathrm{Row}_j(h^*)$ is a subset of $[n+1]^d$ containing exactly $n+1$ indices, defined as those coinciding with the mean index $h^*$ on all coordinates other than the $j$-th one.

\begin{lemma}
For any $j \in [d]$, the $j$-th boundary defined in~\eqref{eq:def-Lj} can be equivalently written as:
    \begin{align*}
      & \mathcal{L}_j = B^{h^{[j]}}_j, 
      & h^{[j]} := \arg\max_{h \in \mathrm{Row}_j(h^*)} B^h_j, 
    \end{align*}
    as long as $\mathcal{Y}^{h^*} \neq \emptyset$, where $h^{[j]} \in \mathrm{Row}_j(h^*)$ satisfies $B^{h^{[j]}}_j> 0$ and $B^{h}_j= 0$ for any $h \in \mathrm{Row}_j(h^*)$ with $h_j \geq h^{[j]}_j$. 
    Moreover, for every $j \in [d]$, if there exists $m\in\mathrm{Row}_j(h^*)$ with $m_j\geq h^*_j$ and $B^m_j=0$
    , then $B^h_j=0$ for all $h\in\mathrm{Row}_j(h^*)$ with $h_j \geq m_j$.
    \label{lem:reduction-search}
\end{lemma}

Under the regularity condition $\mathcal{Y}^{h^*} \neq \emptyset$, which can be immediately verified and typically holds in practice under all but the most extreme scenarios, Lemma~\ref{lem:reduction-search} gives a very fast algorithm for computing $\mathcal{L}_j$ separately for each coordinate $j \in [d]$.
In particular, one can conduct a backward search over $h \in \mathrm{Row}_j(h^*)$ with $h_j = n+1,\ldots,1$, evaluating each $B^h_j$ using~\eqref{eq:rect-wise-bounds} and stopping at the first non-zero value.
Moreover, if $h^{[j]}_j \geq h^*_j$, which can be determined by simply checking whether $B^{h^*}_j > 0$,
then we can find $h^{[j]}_j$ even more effectively through a binary search that checks whether $B^m_j=0$ at each intermediate split step $m \in \mathrm{Row}_j(h^*)$ with $m_j \in [h^*_j,...,n+1]$.

The special case $\mathcal{Y}^{h^*} \neq \emptyset$ corresponds to an extreme situation where even the conservative GWC prediction set $\hat{C}^\mathrm{gwc}(X^{n+1})$ does not cover the empirical mean of the calibration residuals. Typically, this should not occur at moderate miscoverage levels $\alpha$, unless the data are extremely heavy-tailed. Such cases can be handled by simply falling back to the conservative GWC prediction set $\hat{C}^\mathrm{gwc}(X^{n+1})$ itself, which we know to have valid marginal coverage. In conclusion, the procedure described above, summarized in Algorithm~\ref{alg:shortcut}, produces:
\begin{align} \label{eq:tscp-shortcut}
  \hat{C}(X^{n+1}) = \begin{cases}
  \left\{y \in \R^d: 0\leq E_j(X^{n+1},y) \leq \mathcal{L}_j\right\}, &\text{if }\mathcal{Y}^{h^*} \neq \emptyset, \\
    \hat{C}^\mathrm{gwc}(X^{n+1}), &\text{otherwise.}
\end{cases}
\end{align}
\begin{theorem} \label{thm:TSCP-coverage}
Assume $(E^1,\ldots,E^{n+1})$ are exchangeable and Assumption~\ref{eq:assumption-scores} holds.
Then, $\hat{C}(X^{n+1})$ constructed by Algorithm~\ref{alg:shortcut} satisfies: $\mathbb{P}[ Y^{n+1} \in \hat{C}(X^{n+1}) ] \ge 1 - \alpha$.
\end{theorem}
\begin{proof}
Using the characterization of $\hat{C}(X^{n+1})$ in~\eqref{eq:tscp-shortcut}, the proof follows immediately by combining Proposition~\ref{prop:lwc-approximation}, which tells us that $\mathbb{P}[ Y^{n+1} \in \hat{C}^{\mathrm{lwc}}(X^{n+1}) ] \ge 1 - \alpha$, and Theorem~\ref{thm:TSCP-GWC-coverage}, which tells us that $\mathbb{P}[ Y^{n+1} \in \hat{C}^{\mathrm{gwc}}(X^{n+1}) ] \ge 1 - \alpha$.
\end{proof}

Theorem~\ref{thm:TSCP-coverage} guarantees that our prediction sets $\hat{C}(X^{n+1})$ defined in~\eqref{eq:tscp-shortcut} have valid finite-sample marginal coverage above the target level $1-\alpha$.
Moreover, these prediction sets have an interpretable rectangular shape and are typically much tighter than $\hat{C}^\mathrm{gwc}(X^{n+1})$.

\begin{algorithm}[!htbp]
\caption{Transductively Standardized Conformal Prediction (TSCP)}
\label{alg:shortcut}
\textbf{Input:} Calibration residuals $\{E^i\}_{i=1}^{n}$; target miscoverage level $\alpha\in(0,1)$; test input $X^{n+1}$.
\begin{algorithmic}[1]
\State Compute $\hat{\mu}, \hat{\sigma}$ using~\eqref{eq:plug-in-mean-std}.
\State Compute global bounds $\omega_1(\hat{Q}^{\mathrm{gwc}}_{1-\alpha}),\ldots,\omega_d(\hat{Q}^{\mathrm{gwc}}_{1-\alpha})$ using TSCP-GWC. \Comment{Algorithm~\ref{alg:global}}.
\State Sort $\{E^i_j\}_{i=1}^n$ for each $j\in[d]$ to obtain order statistics $E^{(k)}_j$. 
\State Compute $\mathcal{Y}^{h^*}$ using~\eqref{eq:rectangle-explicit}.
\If{$\mathcal{Y}^{h^*}=\emptyset$}
    \State \Return $\displaystyle \hat{C}(X^{n+1}) 
    \gets \left\{y \in \R^d: 0\leq E_j(X^{n+1},y) \leq \omega_j(\hat{Q}^{\mathrm{gwc}}_{1-\alpha}), \quad\forall j \in [d]\right\}$.
\Else
    \For{$j=1$ to $d$}
        \State Compute $\omega_j(\hat{Q}^{\mathrm{lwc}}_{1-\alpha}(h^*))$ using Algorithm~\ref{alg:local-bound} and $B^{h^*}_j$ using~\eqref{eq:rect-wise-bounds}.
        \If{$B^{h^*}_j = 0$}
            \State Backward search over $h_j = h^*_j-1,\ldots, 1$ to find $h^{[j]}$.\hfill
            \Comment{Algorithm~\ref{alg:backward}}
        \Else 
            \State Binary search over $h_j = h^*_j,\ldots, n{+}1$ to find $h^{[j]}$.\hfill
            \Comment{Algorithm~\ref{alg:binary}}
        \EndIf
        \State Update $\mathcal{L}_j \gets B^{h^{[j]}}_j$.
    \EndFor
    \State \Return $\displaystyle \hat{C}(X^{n+1}) 
    \gets \left\{y \in \R^d: 0\leq E_j(X^{n+1},y) \leq \mathcal{L}_j, \quad\forall j \in [d]\right\}$.
\EndIf
\end{algorithmic}
\end{algorithm}

The computational complexity of Algorithm~\ref{alg:shortcut} ranges from $\mathcal{O}(d^2 n \log n)$ in the best case, when binary searches can be used for all coordinates, to $\mathcal{O}(d^2 n^2)$ in the worst case, when all coordinates require backward searches; see Appendix~\ref{app:alg-and-cost} for details. Even in the worst case, a cost of $\mathcal{O}(d^2 n^2)$ is typically not prohibitive in practice: calibration sample sizes $n$ for conformal prediction are usually in the hundreds to low thousands, and beyond that additional samples are often better spent improving model training.

In addition, the computation of the distinct boundaries $\mathcal{L}_1,\ldots,\mathcal{L}_d$ is trivially parallelizable, and the entire procedure need not be repeated across different test inputs, since $\mathcal{L}_1,\ldots,\mathcal{L}_d$ are independent of $X^{n+1}$---remarkably, we have arrived through {\em transductive} logic at a method that is effectively {\em inductive}. Therefore, applying Algorithm~\ref{alg:shortcut} to $m$ different test points using the same calibration data set has total cost ranging from $\mathcal{O}(d^2 n \log n + dm)$ to $\mathcal{O}(d^2 n^2 + dm)$ for typical choices of generalized residuals (Section~\ref{sec:residuals}).

\section{Numerical Experiments}\label{sec:experiments}

\subsection{Methods and Metrics}

We provide a comprehensive empirical validation of the proposed \emph{TSCP} method (Algorithm~\ref{alg:shortcut}) using simulated and real data. 
The performance of TSCP is compared to that of three representative benchmark approaches: (i) \emph{Unscaled Max}, a direct extension of split conformal inference that applies the $\ell_\infty$ norm to the residuals, summarized in Algorithm~\ref{alg:unscaled}; (ii) \emph{Point CHR}, a data-splitting variant recently proposed by \citet{sampson2024conformal}, which estimates per-target out-of-sample points using a held-out calibration set and adjusts them to form axis-aligned rectangular sets; (iii) \emph{Emp.~Copula}, the first distribution-free copula-based method introduced in \citet{messoudi2021copula}. 

The results of additional experiments comparing TSCP to different approaches are in Appendix~\ref{app:additional-experiments}.
We do not include in this empirical study alternative copula-based methods---e.g., vine copulas \citep{pmlr-v258-park25c}---because their main limitations relative to our approach are the same as those of the \emph{Emp.~Copula} benchmark.
To enable meaningful direct comparisons, we also omit approaches that do not produce rectangular prediction sets.

For simplicity, we apply all methods using the same conditional mean-regression model $\hat{f}$ and standard absolute residuals; i.e., $E_j(X,Y)=|\hat{f}_j(X)-Y_j|$ for all $j\in [d]$. Although our method can also be seamlessly applied using different choices of models and residuals, this is a relatively common choice previously employed in several related works \citep{sampson2024conformal, pmlr-v258-park25c, dheur2025unified}.

The conformal prediction sets produced by the different methods are compared based on two key performance metrics: average coverage and size (volume). Since all methods rely on the same underlying model, it suffices to measure the size of the prediction sets in residual space. Specifically, we report
\begin{align*}
    \text{(Test) Coverage} &:= \frac{1}{N}\sum^N_{l=1} \frac{\text{\# of included test points}}{\text{\# of total test points}},\\
    \text{(Residual space) Volume} &:= \frac{1}{N}\sum^N_{l=1} \frac{1}{2^d} \text{Vol}(\text{prediction set}),
\end{align*}
where both metrics are averaged over $N$ repetitions of each experiment, corresponding to independent data sets. The constant factor $2^{-d}$ in the volume definition facilitates comparison across experiments with different output dimensions.

\subsection{Experiments with Simulated Data}

We simulate data with features from a $10$-dimensional standard normal distribution with independent components, and 10-dimensional outcomes
\begin{equation}
    Y_j \mid X \sim f_j(X) + \epsilon_j, \quad \forall j \in [10],
    \label{eq:data-generator}
\end{equation}
where $f_j(X) = \sum^{10}_{i=1}\xi_i X_i$ such that $\xi_1,\ldots,\xi_{10} \sim \text{Unif}(-10, 10)$, and $\epsilon_j$ is independent random noise for each $j \in [d]$.
We consider six settings corresponding to different noise distributions; the results for two are presented here while those corresponding to other settings, which lead to qualitatively similar conclusions, are summarized in Appendix~\ref{app:simulations}.

In each setting, we simulate a training data set containing $7200$ observations to fit a linear regression model $\hat{f}$ using ordinary least squares, a calibration set $\dataset_{\mathrm{cal}}$ of size $|\dataset_{\mathrm{cal}}| \in \{30, 50, 100, 300, 500\}$, and a test set $\dataset_{\mathrm{test}}$ of size $800$, to construct prediction sets at level $\alpha=0.1$. All results are averaged over $N=200$ independent repetitions.

Figure~\ref{fig:gaussian-heter} summarizes the results obtained in the first setting, where the noise distribution is Gaussian with heterogeneous variance across dimensions:
\[
\epsilon^{\mathrm{heter}}_j \sim \mathcal{N}\left(0,\; (10 - j + 1)^2\right), \quad \forall j \in [10].
\]
\begin{figure}[!htbp]
    \centering
    \includegraphics[width=0.9\textwidth]{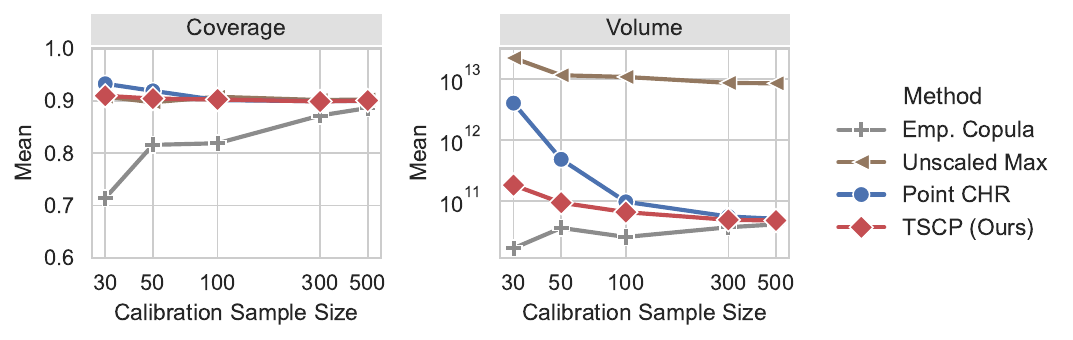} 
\caption{Average performance of the proposed TSCP method and alternative conformal prediction approaches on simulated data with $d=10$ outcome variables and heterogeneous noise variances. The target coverage level is $90\%$. TSCP consistently attains near-nominal coverage while producing the smallest prediction sets on average. Numerical results and corresponding standard deviations are reported in Table~\ref{tab:gaussian-heter}. }
     \label{fig:gaussian-heter}
\end{figure}

These results show that our method performs well across both small and large sample sizes, always achieving valid coverage at level $1-\alpha$ (as predicted by the theory) but not excessively above it, using relatively tight prediction sets.
The \emph{Point CHR} approach yields substantially larger, and therefore less informative, prediction sets—especially in small samples (e.g., $n < 100$)—due to the inefficiency introduced by its additional sample split. Even for moderately large samples, our method continues to perform noticeably better, although this difference is difficult to discern in Figure~\ref{fig:gaussian-heter} because of the log–log scale; see Table~\ref{tab:gaussian-heter} for a more detailed numerical summary.
The \emph{Emp.~Copula} approach clearly fails to achieve finite-sample coverage, since it uses the same calibration set both to estimate the copula model and to calibrate the prediction sets based on the resulting nonconformity scores. While this issue could be addressed through additional data splitting, doing so would come with the same efficiency disadvantages of \emph{Point CHR}. This limitation is shared by related copula-based approaches that we do not explicitly include in our comparisons.
Finally, the \emph{Unscaled Max} approach produces much larger (and effectively uninformative) prediction sets in both small and large samples, as it is negatively affected by the unequal variances of the noise distribution across different dimensions.

Figure~\ref{fig:gaussian-homo} reports results from analogous experiments in which the noise distribution has equal variance across dimensions, with $\epsilon^{\mathrm{homo}}_j \sim \mathcal{N}(0, 1)$ for all $j \in [10]$. This setting strongly favors the \emph{Unscaled Max} approach; nevertheless, our TSCP still performs as a close second.
In practice, however, noise levels are rarely known \emph{a priori}, making the adaptivity and robustness of TSCP generally desirable.

\begin{figure}[!htbp]
    \centering
    \includegraphics[width=0.9\textwidth]{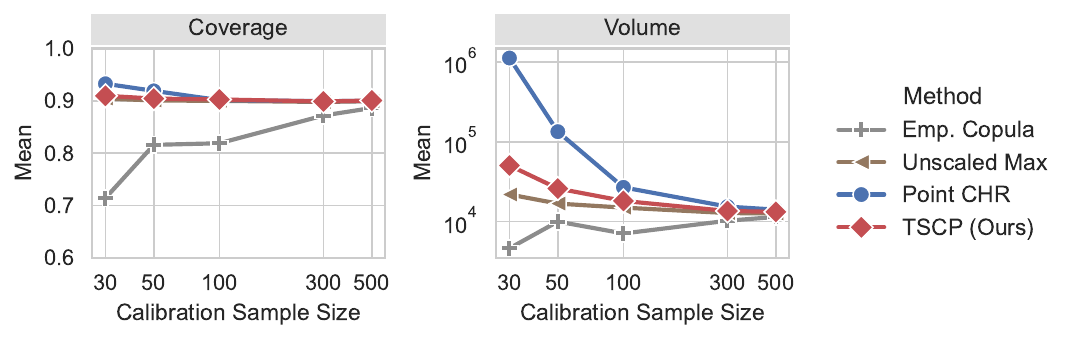}
\caption{Average performance of the proposed TSCP method and alternative conformal prediction approaches in numerical experiments similar to those of Figure~\ref{fig:gaussian-heter}. Here, the outcome variables have homogeneous variance across their $d=10$ dimensions, which gives an advantage to the \emph{Unscaled Max} approach. The target coverage level is $0.9$. TSCP consistently attains near-nominal coverage while producing the smallest prediction sets on average. Numerical results and corresponding standard deviations are reported in Table~\ref{tab:gaussian-homo}.}
     \label{fig:gaussian-homo}
\end{figure}

The results of additional numerical experiments are reported in Appendices~\ref{app:ours_simulation}--\ref{app:hevy_tailed_simulation}. In Appendix~\ref{app:ours_simulation}, we compare TSCP with several alternative approaches described in Sections~\ref{sec:preliminaries} and~\ref{sec:methods}, including the naive plug-in rescaling procedure (Algorithm~\ref{alg:std-plug-in}), the data-splitting variant (Algorithm~\ref{alg:std-ds}), the conservative global worst-case approach (Algorithm~\ref{alg:global}), and the union-based construction without rectangular enclosure (Section~\ref{sec:method-lwc}). These experiments highlight the favorable balance between statistical–computational trade-offs achieved by TSCP. We further examine scalability with respect to the outcome dimension and show that Algorithm~\ref{alg:shortcut} remains computationally efficient in regimes where union-based constructions become prohibitive. In Appendix~\ref{app:hevy_tailed_simulation}, we investigate the robustness of TSCP under heavy-tailed noise and demonstrate that our method continues to achieve near-nominal coverage even in extreme settings where the population variance of the residuals is infinite.

\subsection{Experiments with Real Data}

We now evaluate all methods on six real datasets. 
The first four, \emph{rf2} \citep{grigorios_tsoumakas_eleftherios_spyromitros-xioufis_william_groves_ioannis_vlahavas_2022}, \emph{scm1d} \citep{scm1d}, \emph{scm20d} \citep{scm20d}, and \emph{stock} \citep{stock_portfolio_performance_390}, are borrowed from related prior works \citep{pmlr-v258-park25c, dheur2025unified}. 
The other two datasets, \emph{student} \citep{cortez2008student} and \emph{energy} \citep{tsanas2012energy}, are obtained from the UCI repository and serve to illustrate the small-sample performance of TSCP.

For each dataset, we perform $200$ random train/calibration/test splits using a $75\%/5\%/20\%$ ratio. We fit a multivariate sparse (lasso) regression model with regularization parameter $\lambda = 0.0001$ on \emph{stock}, and random forest regression models on all other datasets; see the GitHub software repository linked in the Software Availability section for additional reproducibility details.
Table~\ref{tab:real} reports the average coverage and volume, along with one standard deviation.
As shown in Table~\ref{tab:real}, the empirical behavior of all methods on real data largely mirrors the trends observed in the simulation studies. Among methods with theoretical coverage guarantees, TSCP achieves the smallest prediction-set volume on the \emph{stock}, \emph{scm1d}, \emph{scm20d}, and \emph{energy} datasets. On \emph{rf2}, TSCP is outperformed by \emph{Point CHR}, likely due to the presence of a small number of influential outliers. On \emph{student}, \emph{Unscaled Max} yields smaller volumes, consistent with the approximately homogeneous noise across targets (student GPAs at different grade levels).
The \emph{Point CHR} approach produces an uninformative, \emph{infinite} prediction set on the \emph{stock} dataset, because its additional data split makes the effective calibration set extremely small.
Finally, \emph{Emp.~Copula} generally fails to achieve the desired coverage level, consistent with its lack of finite-sample guarantees.
Overall, the results show that TSCP provides a reasonable trade-off between coverage and volume, excelling especially when calibration data are limited or when noise scales differ significantly across targets.

\begin{table}[!htbp]
\input{figures/realdata_table_highlighted}
\caption{Performance of the proposed TSCP method and alternative conformal prediction approaches on real datasets. Results are averaged over $200$ random splits of the data into training/calibration/and test subsets, with one standard deviation reported in parentheses. The target coverage level is $0.9$. Coverage values below the target are highlighted in \textcolor{red}{red}, and the smallest volume among methods achieving valid coverage is highlighted in \textbf{bold}.}
\label{tab:real}
\end{table}

To emphasize their interpretability, Table~\ref{tab:toy} displays the rectangular prediction sets produced by different methods for two representative test points from the \emph{energy} dataset, where $d=2$. In this setting, rectangular prediction sets consist of pairs of simultaneous prediction intervals that aim to jointly contain the true outcomes with probability $90\%$.

\begin{table}[!htbp]
\input{figures/toy_real}
\caption{Anecdotal illustration of conformal prediction sets of rectangular shape for two test points from the \emph{energy} dataset. In this case, the prediction sets can be reported as pairs of joint prediction intervals for the two outcome dimensions. Volumes are computed in residual space. Intervals that fail to cover the corresponding true outcomes are highlighted in \textcolor{red}{red}; smallest interval that covers the ground truth is highlighted in \textbf{bold}. This example is included for illustration only: the coverage guarantees of conformal prediction hold on average, not for individual test points.}
\label{tab:toy}
\end{table}

\section{Discussion and Future Work}
\label{sec:discussion}

The core idea underlying the method presented in this paper is the explicit standardization of residuals, which contributes both to the transparency and to the computational efficiency of the resulting prediction sets. As with other standardization-based methods, however, this strategy introduces some sensitivity to outliers. This behavior is visible in the numerical experiments based on the \emph{rf2} dataset, where a small number of influential observations can substantially expand the resulting prediction sets. An important direction for future work is therefore the development of robust extensions of our method, for example through rescaling based on median absolute deviation or quantile-based normalizations, while preserving the exchangeability required for conformal validity.

Additional directions for future work include allowing for negative residuals. Such an extension could, for example, enable shrinkage of baseline prediction sets when starting from quantile regression models that already exhibit conservative coverage. It would also be of interest to adapt the proposed approach to settings with non-numerical outcomes, such as categorical or ordinal responses. Finally, our method could be extended to handle partly observed outcomes \citep{braun2025multivariate}.

\section*{Software Availability}

A software implementation of the methods described in this paper, along with code necessary to reproduce the numerical experiments, is available online at \url{https://github.com/OTheMagic/multi-target-scaling.git}.

\acks{The authors thank Rundong Ding and Dr.~Yizhe Zhu for helpful suggestions during an oral presentation of an earlier version of this work. M.S.~was partly supported by a NSF grant DMS 2210637, a USC-Capital One CREDIF award, and a Google Research Scholar award.}


%% file: figures/realdata_table_highlighted.tex
\centering
\captionsetup[subtable]{justification=centering}
\small
\begin{subtable}{\textwidth}
\centering
\begin{tabular}[t]{l c c c c }
\toprule
\textbf{Method} & \multicolumn{2}{c}{Stock $(d=6,\, n=15)$} & \multicolumn{2}{c}{rf2 $(d=8,\, n=384)$} \\
\cmidrule(lr){2-3}\cmidrule(lr){4-5}
& Coverage & Volume & Coverage & Volume \\
\midrule
Unscaled Max & 0.940(0.057) & 6.03e-02(1.57e-01) & 0.900(0.016) & 4.14e+03(2.90e+03) \\
Emp. copula & \color{red}0.727(0.113) & 1.01e-05(9.80e-06) & \color{red}{0.893(0.016)} & 5.35e+01(4.69e+01) \\
Point CHR & 1.000(0.000) & \text{inf} & 0.901(0.022) & \textbf{{6.85e+01(7.55e+01)}} \\
TSCP & 0.956(0.059) & \textbf{{1.86e-03(8.55e-03)}} & 0.901(0.017) & 5.42e+02(1.50e+03) \\
\midrule

\end{tabular}
\end{subtable}

\small
\begin{subtable}{\textwidth}
\centering
\begin{tabular}[t]{l c c c c }
\midrule
\textbf{Method} & \multicolumn{2}{c}{scm1d $(d=16,\, n=490)$} & \multicolumn{2}{c}{scm20d $(d=16,\, n=448)$} \\
\cmidrule(lr){2-3}\cmidrule(lr){4-5}
& Coverage & Volume & Coverage & Volume \\
\midrule
Unscaled Max & 0.900(0.016) & 2.74e+39(4.02e+39) & 0.901(0.016) & 2.70e+40(2.92e+40) \\
Emp. copula & \color{red}{0.893(0.016)} & 1.09e+39(1.34e+39) & \color{red}{0.892(0.017)} & 1.178e+40(1.190e+40) \\
Point CHR & 0.905(0.020) & 2.84e+39(5.43e+39) & 0.902(0.023) & 3.00e+40(8.00e+40) \\
TSCP & 0.901(0.015) & \textbf{{1.52e+39(2.61e+39)}} & 0.901(0.016) & \textbf{{1.53e+40(1.50e+40)}} \\
\midrule

\end{tabular}
\end{subtable}

\small
\begin{subtable}{\textwidth}
\centering
\begin{tabular}[t]{l c c c c }
\midrule

\textbf{Method} & \multicolumn{2}{c}{Energy $(d=2,\, n=38)$} & \multicolumn{2}{c}{Student $(d=3,\, n=32)$} \\
\cmidrule(lr){2-3}\cmidrule(lr){4-5}
& Coverage & Volume & Coverage & Volume \\
\midrule
Unscaled Max & 0.917(0.049) & 1.58e+01(6.80e+00) & 0.904(0.055) & \textbf{{1.51e+02(1.27e+02)}} \\
Emp. copula & \color{red}0.886(0.061) & 5.41e+00(4.49e+00) & \color{red}0.861(0.073) & 1.08e+02(7.19e+01) \\
Point CHR & 0.899(0.069) & 8.81e+00(2.42e+01) & 0.939(0.057) & 6.89e+02(1.07e+03) \\
TSCP & 0.922(0.048) & \textbf{{6.95e+00(4.43e+00)}} & 0.912(0.056) & 1.98e+02(1.77e+02) \\

\bottomrule
\end{tabular}
\end{subtable}

%% file: figures/toy_real.tex
\centering
\captionsetup[subtable]{justification=centering}
\small
\begin{subtable}{\linewidth}
\centering
\begin{tabular}[t]{l c c c c}
\toprule
\textbf{Test 1/Methods} 
& Outcome 1 & Outcome 2 & Volume & Ground Truth \\
\midrule
Unscaled Max & [31.70, 38.94] & \textbf{[34.74, 41.98]} & 13.09 & \multirow{4}{*}{(35.40, 39.22)} \\
Emp. copula & \textbf{[34.36, 36.28]} & \textbf{[34.74, 41.98]} & {3.47} &\\
Point CHR & [34.29, 36.36] & [31.45, 45.27] & 7.14 &\\
TSCP & [34.28, 36.36] & [33.63, 43.08] & 4.93 &\\
\midrule
\end{tabular}
\end{subtable}

\small
\begin{subtable}{\linewidth}
\centering
\begin{tabular}[t]{l c c c c}
\midrule
\textbf{Test 2/Methods} 
& Outcome 1 & Outcome 2 & Volume & Ground Truth \\
\midrule
Unscaled Max & [25.44, 32.67] & \color{red}[29.34, 36.58] & 13.09 & \multirow{4}{*}{(29.88, 28.31)} \\
Emp. copula & \textbf{[28.10, 30.02]} & \color{red}[29.34, 36.58] & 3.47 & \\
Point CHR & [28.02, 30.09] & [26.05, 39.87] & 7.14 & \\
TSCP & [28.01, 30.10] & \textbf{[28.24, 37.69]} & {4.93} & \\
\bottomrule
\end{tabular}
\end{subtable}

%% file: supplementary.tex
\setcounter{algorithm}{0}
\setcounter{figure}{0}
\setcounter{equation}{0}

\renewcommand{\thetheorem}{S\arabic{theorem}}
\renewcommand{\thelemma}{S\arabic{lemma}}
\renewcommand{\theequation}{S\arabic{equation}}
\renewcommand{\thefigure}{S\arabic{figure}}
\renewcommand{\thetable}{S\arabic{table}}
\renewcommand{\thealgorithm}{S\arabic{algorithm}}

\begin{center}
    \Large \textbf{Appendix: Algorithms, Proofs, and Numerical Results}
\end{center}

\section{Additional Algorithmic Details} \label{app:algorithms}

\subsection{Implementation of Alternative Methods}

\begin{algorithm}[!htb]
    \caption{Separate Prediction Intervals with Bonferroni Correction}
    \label{alg:bonferroni}
    \textbf{Input:} Calibration residuals $\{E^i\}_{i=1}^{n}$; target miscoverage level $\alpha\in(0,1)$; test input $X^{n+1}$.
    \begin{algorithmic}[1]
        \State Compute $\alpha^* \gets \alpha/d$.
        \For{$j=1$ to $d$}
        \State Compute $\hat{Q}_{1-\alpha^*,j} \gets \hat{Q}_{1-\alpha^*}(E^1_j,\ldots,E^n_j)$.
        \EndFor
        \State \Return $\hat{C}^{\mathrm{Bonf.}}(X^{n+1}) \gets \left\{y \in \R^d: 0 \leq E_j(X^{n+1},y) \leq \hat{Q}_{1-\alpha^*,j}, \forall j \in [d]\right\}$.
    \end{algorithmic}
\end{algorithm}

\begin{algorithm}[!htb]
    \caption{Prediction Intervals via Dimension Reduction with $\ell_\infty$-Norm (Unscaled)}
    \label{alg:unscaled}
    \textbf{Input:} Calibration residuals $\{E^i\}_{i=1}^{n}$; target miscoverage level $\alpha\in(0,1)$; test input $X^{n+1}$.
    \begin{algorithmic}[1]
        \State Compute $S^i_{\mathrm{unscaled}} \gets \|E^i\|_\infty$ for $i = 1,\ldots,n$.
        \State Compute $\hat{Q}^{\mathrm{unscaled}}_{1-\alpha} \gets \hat{Q}_{1-\alpha}(S^1_{\mathrm{unscaled}},\dots, S^n_{\mathrm{unscaled}})$
        \State \Return $\hat{C}^{\mathrm{unscaled}}(X^{n+1}) \gets \left\{y \in \R^d: 0 \leq E^{n+1}_j(y) \leq \hat{Q}^{\mathrm{unscaled}}_{1-\alpha}, \forall j \in [d]\right\}$.
    \end{algorithmic}
\end{algorithm}

\begin{algorithm}[!htb]
    \caption{Heuristic Plug-In Standardization}
    \label{alg:std-plug-in}
    \textbf{Input:} Calibration residuals $\{E^i\}_{i=1}^{n}$; target miscoverage level $\alpha\in(0,1)$; test input $X^{n+1}$.
    \begin{algorithmic}[1]
        \For{$j=1$ to $d$}
            \State Compute the location and scale parameters
            \begin{align*}
            & \hat{\mu}^{\mathrm{pi}}_j \gets \frac{1}{n}\sum_{i=1}^{n} E^i_j,
              & \hat{\sigma}^{\mathrm{pi}}_j \gets \sqrt{\frac{1}{n-1}\sum_{i=1}^{n}(E^i_j-\hat{\mu}^{\mathrm{pi}}_j)^2}.
            \end{align*}
            \EndFor
        \State Compute $S^i_{\mathrm{pi}} \gets \Phi(E^i; \hat{\mu}_j^\mathrm{pi},\hat{\sigma}_j^\mathrm{pi})$ for $i \in [n]$, using $\Phi$ defined in~\eqref{eq:residual-transformation}.
        \State Compute $\hat{Q}^\mathrm{pi}_{1-\alpha}\gets \hat{Q}_{1-\alpha}(S^1_\mathrm{pi},\dots, S^n_\mathrm{pi}).$
        \State Compute 
            \begin{align*}
\hat{C}^{\mathrm{pi}}(X^{n+1}) \gets \left\{y \in \R^d: 0 \leq E_j(X^{n+1}, y) \leq \hat{Q}^\mathrm{pi}_{1-\alpha}\cdot \hat{\sigma}_j^\mathrm{pi} + \hat{\mu}_j^\mathrm{pi}, \forall j \in [d]\right\},
            \end{align*}
            using either~\eqref{eq:pred-abs-res} (for absolute regression residuals) or~\eqref{eq:pred-abs-qr} (for quantile regression residuals).
            \State \Return Prediction set $\hat{C}^{\mathrm{pi}}(X^{n+1})$.
    \end{algorithmic}
\end{algorithm}

\begin{algorithm}[!htb]
    \caption{Standardization with Data Splitting}
    \label{alg:std-ds}
    \textbf{Input:} Calibration residuals $\{E^i\}_{i=1}^{n}$; target miscoverage level $\alpha\in(0,1)$; test input $X^{n+1}$.\\
    \textbf{Output:} $\hat{C}^{\mathrm{ds}}(X^{n+1})$
    \begin{algorithmic}[1]
        \State Randomly split $\{1,\ldots,n\}$ into two disjoint subsets, $\mathcal{I}_1$ and $\mathcal{I}_2$.
        \For{$j=1$ to $d$}
            \State Compute the location and scale parameters
            \begin{align*}
            & \hat{\mu}^{\mathrm{ds}}_j \gets \frac{1}{|\mathcal{I}_1|}\sum_{i \in I_1}E^i_j,
              & \hat{\sigma}^{\mathrm{ds}}_j \gets \sqrt{\frac{1}{|\mathcal I_1|-1}\sum_{i \in I_1}(E^i_j-\hat{\mu}_j)^2}.
            \end{align*}
        \EndFor
        \State Compute $S^i_{\mathrm{ds}} \gets \Phi(E^i; \hat{\mu}_j^\mathrm{ds},\hat{\sigma}_j^\mathrm{ds})$ for $i \in \mathcal I_2$, using $\Phi$ defined in~\eqref{eq:residual-transformation}.
        \State Compute $\hat{Q}^\mathrm{ds}_{1-\alpha}\gets \hat{Q}_{1-\alpha}(S^1_\mathrm{ds},\dots, S^n_\mathrm{ds}).$
        \State \Return $\hat{C}^{\mathrm{ds}}(X^{n+1}) \gets \left\{y \in \R^d: 0 \leq E_j(X^{n+1}, y) \leq \hat{Q}^\mathrm{ds}_{1-\alpha}\cdot \hat{\sigma}_j^\mathrm{ds} + \hat{\mu}_j^\mathrm{ds}, \forall j \in [d]\right\}$.
    \end{algorithmic}
\end{algorithm}

\FloatBarrier

\subsection{Auxiliary Algorithms}

\begin{algorithm}[!htb]
\caption{Global-Worst-Case TSCP (TSCP-GWC)}
\label{alg:global}
\textbf{Input:} Calibration residuals $\{E^i\}_{i=1}^{n}$; target miscoverage level $\alpha\in(0,1)$; test input $X^{n+1}$.
\begin{algorithmic}[1]
\State Compute $\hat\mu, \hat{\sigma}$ using $~\eqref{eq:plug-in-mean-std}$.
\State Compute $S_{\mathrm{gwc}}^i\gets \Phi^{\mathrm{gwc}}(E^i)$ using $\Phi^{\mathrm{gwc}}$ defined in Lemma~\ref{lem:gwc-rescaling}.
\State Compute $\hat{Q}^{\mathrm{gwc}}_{1-\alpha} \gets \hat{Q}_{1-\alpha}(S_1^{\mathrm{gwc}},\dots, S^n_{\mathrm{gwc}})$.
\State Compute $\omega_1(\hat{Q}^{\mathrm{gwc}}_{1-\alpha}),\ldots,\omega_d(\hat{Q}^{\mathrm{gwc}}_{1-\alpha})$ using $\omega$ defined in~\eqref{eq:key-ineq}.
\State \Return $\displaystyle \hat{C}^{\mathrm{gwc}}(X^{n+1}) \gets \left\{y\in \R^d: 0\leq E_j(X^{n+1}, y) \leq  \omega_j(\hat{Q}^{\mathrm{gwc}}_{1-\alpha}), \forall j \in [d]\right\}$.
\end{algorithmic} 
\end{algorithm}

\begin{algorithm}[!htbp]
\caption{Backward Search in Algorithm~\ref{alg:shortcut}}
\label{alg:backward}
\textbf{Input:} Calibration residuals $\{E^i\}_{i=1}^{n}$; global bounds $\omega_1(\hat{Q}^{\mathrm{gwc}}_{1-\alpha}),\ldots,\omega_d(\hat{Q}^{\mathrm{gwc}}_{1-\alpha})$; target miscoverage level $\alpha\in(0,1)$; calibration mean $\hat{\mu}$ and standard deviation\ $\hat{\sigma}$; mean index $h^*$; current coordinate $j$.
\begin{algorithmic}[1]
    \For{$h\in \mathrm{Row}_j(h^*)$ with $h_j=h^*_j-1$ to $1$}
        \State Compute $\omega_j(\hat{Q}^{\mathrm{lwc}}_{1-\alpha}(h))$ using Algorithm~\ref{alg:local-bound} and $B^{h}_j$ using~\eqref{eq:rect-wise-bounds}.
        \If{$B^{h}_j > 0$}
            \State \Return $h^{[j]} = h$.
        \EndIf
    \EndFor
\end{algorithmic}
\end{algorithm}

\begin{algorithm}[!htb]
\caption{Binary Search in Algorithm~\ref{alg:shortcut}}
\label{alg:binary}
\textbf{Input:} Calibration residuals $\{E^i\}_{i=1}^{n}$ global bounds $\omega_1(\hat{Q}^{\mathrm{gwc}}_{1-\alpha}),\ldots,\omega_d(\hat{Q}^{\mathrm{gwc}}_{1-\alpha})$; ; target miscoverage level $\alpha\in(0,1)$; calibration mean $\hat{\mu}$ and standard deviation\ $\hat{\sigma}$; mean index $h^*$; current coordinate $j$.
\begin{algorithmic}[1]
    \State Let $I_{\min} \gets h^*_j$ and $I_{\max} \gets n+1$.
    \While{$I_{\min} < I_{\max}$}
        \State Compute $m_j\gets \lceil (I_{\mathrm{min}}+I_\mathrm{max})/2\rceil$ and $m_k = h^*_k$ for all $k\neq j$.
        \State Compute $\omega_j(\hat{Q}^{\mathrm{lwc}}_{1-\alpha}(h))$ and $B^{m}_j$ using~\eqref{eq:rect-wise-bounds}.\hfill
        \Comment{Algorithm~\ref{alg:local-bound}}
        \If{$B^{m}_j > 0$}
            \State $I_{\mathrm{min}} \gets m_j$.
        \Else 
            \State $I_{\mathrm{max}} \gets m_j-1$.
        \EndIf
        \If{$I_{\max} = I_{\min}$}
            \State \Return $h^{[j]}_j = I_{\min}$ and $h^{[j]}_k = h^*_k$ for all $k\neq j$.
        \EndIf
    \EndWhile
\end{algorithmic}
\end{algorithm}

\FloatBarrier

\subsection{Computational Cost Analysis}\label{app:alg-and-cost}

\subsubsection{Algorithm~\ref{alg:global}}
\label{app:gwc-cost}
Since the calibration residuals $\{E^i\}_{i=1}^{n}$ are stored as an $(n,d)$-array, we know
\begin{enumerate}
    \item Computing $(\hat\mu_j,\hat\sigma_j)_{j=1}^d$using~\eqref{eq:plug-in-mean-std} requires $\mathcal O(nd)$.
    \item Evaluating $S_{\mathrm{gwc}}^i\gets \Phi^{\mathrm{gwc}}(E^i)$ for all $i\in[n]$ is also $\mathcal O(nd)$.
    \item Computing the quantile $\hat{Q}^{\mathrm{gwc}}_{1-\alpha}$ costs $\mathcal O(n)$ using a selection algorithm such as quickselect. 
    \item Computing $\omega_1(\hat{Q}^{\mathrm{gwc}}_{1-\alpha}),\ldots,\omega_d(\hat{Q}^{\mathrm{gwc}}_{1-\alpha})$ costs $\mathcal{O}(d)$ once $\hat{Q}^{\mathrm{gwc}}_{1-\alpha}$ and $(\hat\mu_j,\hat\sigma_j)_{j=1}^d$ are known.
\end{enumerate}
Hence, the overall computational complexity of Algorithm~\ref{alg:global} is $\mathcal O(nd)$.

\subsubsection{Algorithm~\ref{alg:local-bound}}
\label{app:lwc-local-bound}

Algorithm~\ref{alg:local-bound} shares the same structure and computational cost as Algorithm~\ref{alg:global}: $\mathcal O(nd)$.

\subsubsection{Algorithm~\ref{alg:shortcut}}
\label{app:lwc-sc-cost}
Since the calibration residuals $\{E^i\}_{i=1}^{n}$ are stored as an $(n,d)$-array, we know
\begin{enumerate}
    \item Computing $(\hat\mu_j,\hat\sigma_j)_{j=1}^d$ using~\eqref{eq:plug-in-mean-std} and $\omega_1(\hat{Q}^{\mathrm{gwc}}_{1-\alpha}),\ldots,\omega_d(\hat{Q}^{\mathrm{gwc}}_{1-\alpha})$ using Algorithm~\ref{alg:global} requires $\mathcal O(nd)$.
    \item Sorting $\{E^i_j\}^n_{i=1}$ for all $j\in [d]$ costs $\mathcal{O}(dn\log{n})$;
    \item Computing $\mathcal{Y}^{h^*}$ relies on computing $h^*$, which requires masking all values below $\hat{\mu}$ and report the index of the lowest value from the sorted residuals list $\{E^i_j\}^n_{i=1}$, each of which costs $\mathcal{O}(n)$, so together this operation contributes $\mathcal{O}(dn)$.
    \item If $\mathcal{Y}^{h^*} = \emptyset$, then the cost is only $\mathcal{O}(d)$ since all global bounds have been computed.
    \item If $\mathcal{Y}^{h^*} \neq \emptyset$, then for each coordinate, we compute the local bounds using Algorithm~\ref{alg:local-bound} and $B^{h^*}_j$ using~\eqref{eq:rect-wise-bounds}, which costs $\mathcal{O}(dn)$.
    If $B^{h^*}_j = 0$, then we use the backward search detailed in Algorithm~\ref{alg:backward}, which costs $\mathcal{O}(dn^2)$; see Appendix~\ref{app:backward-backward}.
    If $B^{h^*}_j > 0$, then we use the binary search detailed in Algorithm~\ref{alg:binary}, which costs $\mathcal{O}(dn\log{n})$; see Appendix~\ref{app:backward-binary}.
\end{enumerate}  
Let $T \leq d$ be the number of coordinates where binary searches are enabled, then the total cost of Algorithm~\ref{alg:shortcut} is $\mathcal{O}(Tdn\log{n} + (d-T)dn^2)$

\subsubsection{Algorithm~\ref{alg:backward}} \label{app:backward-backward}
Running Algorithm~\ref{alg:local-bound} and forming $B^h_j$ costs $\mathcal{O}(dn)$ for each $h$; there are $n+1$ indices in
$\mathrm{Row}_j(h^*)$, so the computational complexity of Algorithm~\ref{alg:backward} never exceeds $\mathcal{O}(dn^2)$.

\subsubsection{Algorithm~\ref{alg:binary}} \label{app:backward-binary}
Running Algorithm~\ref{alg:local-bound} and forming $B^h_j$ costs $\mathcal{O}(dn)$ for each $h$. The number of indices $h$ encountered in this algorithm is $\mathcal{O}(\log_2{n})$ because we split and continue in only half of the current indices sequence at each step. This is a standard binary search. Therefore, the total computational complexity of Algorithm~\ref{alg:binary} never exceeds $\mathcal{O}(dn\log{n})$.

\clearpage

\section{Mathematical Proofs} 
\label{app:proofs}

\subsection{Auxiliary Lemmas}\label{app:proofs-lemmas}

\begin{lemma}\label{lem:1}
    Let $x^1,\ldots,x^n \in \R$ and $\bar{x}^n = \frac{1}{n}\sum^n_{i=1}x^i$. For any $x^{n+1}\in \R$, let $\hat\sigma(x^{n+1})$ be the sample standard deviation of $\{x^1,\ldots,x^n, x^{n+1}\}$, then
    \[
    \hat\sigma^2(x^{n+1}) = \frac{(x^{n+1}-\bar{x}^{n})^2}{n+1}+\frac{1}{n}\sum^n_{i=1}(x^i-\bar{x}^{n})^2.
    \]
\end{lemma}
\begin{proof}
    Denote $\bar{x}^{n+1} = \frac{1}{n+1}\sum^{n+1}_{i=1}x^i$. We notice that
\begin{align*}
    \hat\sigma^2(x^{n+1}) = &\frac{1}{n}\sum^{n+1}_{i=1}(x^i-\bar{x}^{n+1})^2\\
    = &\frac{1}{n} \left[ (x^{n+1}-\bar{x}^{n+1})^2+\sum^n_{i=1}(x^i-\bar{x}^{n}+\bar{x}^{n}-\bar{x}^{n+1})^2 \right]\\
    = &\frac{(x^{n+1}-\bar{x}^{n+1})^2}{n}+\frac{1}{n}\sum^n_{i=1}\left[(x^i-\bar{x}^{n})^2+2(x^i-\bar{x}^n)(\bar{x}^n-\bar{x}^{n+1})+(\bar{x}^{n}-\bar{x}^{n+1})^2\right]\\
    = &\frac{(x^{n+1}-\bar{x}^{n+1})^2}{n}+(\bar{x}^{n}-\bar{x}^{n+1})^2+\frac{1}{n}\sum^n_{i=1}(x^i-\bar{x}^{n})^2.
\end{align*}
Substituting $\bar{x}^{n+1} = \frac{x^{n+1}+n\bar{x}^n}{n+1}$, we get
\begin{align*}    
    \hat\sigma^2(x^{n+1}) = &\frac{1}{n}\left(x^{n+1}-\frac{x^{n+1}+n\bar{x}^n}{n+1}\right)^2+\left(\bar{x}^n-\frac{n\bar{x}^n+x^{n+1}}{n+1}\right)^2+\frac{1}{n}\sum^n_{i=1}(x^i-\bar{x}^{n})^2\\
    = &\frac{n}{(n+1)^2}(x^{n+1}-\bar{x}^n)^2+\frac{1}{(n+1)^2}(\bar{x}^n-x^{n+1})^2+\frac{1}{n}\sum^n_{i=1}(x^i-\bar{x}^{n})^2\\
    = &\frac{(x^{n+1}-\bar{x}^{n})^2}{n+1}+\frac{1}{n}\sum^n_{i=1}(x^i-\bar{x}^{n})^2.
\end{align*}
\end{proof}

\begin{lemma}\label{lem:2}
    Let $X^1,\ldots, X^n \in \R$ be real-valued random variables with sample mean $\bar{X}^n$. Let $X^{n+1}$ and $\tilde{X}^{n+1}$ be another two random variables such that
    \[
    |X^{n+1} - \bar{X}^n| \leq |\tilde{X}^{n+1}-\bar{X}^n| \quad \text{a.s.}
    \]
    Then,
    \[
    \hat{\sigma}(X^{n+1}) \leq \hat{\sigma}(\tilde{X}^{n+1}) \quad \text{a.s.},
    \]
    where $\hat{\sigma}(X^{n+1})$ is the the sample standard deviation of $\{X^1,\ldots,X^n, X^{n+1}\}$, as in Lemma~\ref{lem:1}.
\end{lemma}
\begin{proof}
    This is a direct application of Lemma~\ref{lem:1}:
\begin{align*}
    \hat{\sigma}(X^{n+1}) 
    &= \frac{(X^{n+1}-\bar{X}^{n})^2}{n+1}+\frac{1}{n}\sum^n_{i=1}(X^i-\bar{X}^{n})^2\\
    &\leq \frac{(\tilde{X}^{n+1}-\bar{X}^{n})^2}{n+1}+\frac{1}{n}\sum^n_{i=1}(X^i-\bar{X}^{n})^2 \quad \text{[Assumption]}\\
    &=\hat{\sigma}(\tilde{X}^{n+1}).
\end{align*}
\end{proof}

\begin{lemma}\label{lem:3}
    Let $X^1,\ldots, X^{n+1} \in \R$ be real-valued random variables with sample mean $\bar{X}^{n+1}$ and sample standard deviation $\bar{S}^{n+1} > 0$. Define $Z^i = (X^i-\bar{X}^{n+1}) / \bar{S}^{n+1}$ for $i \in [n+1]$. Then,
    \[
    |Z^i| \leq \frac{n}{\sqrt{n+1}} \quad\text{almost surely}, \quad \forall i \in [n+1].
    \]
    This inequality is strict if $X^1,\ldots,X^{n+1}$ are almost-surely distinct.
\end{lemma}
\begin{proof}
Fix any $i\in [n+1]$.
Let $\bar{X}^{-i} = \frac{1}{n}\sum_{k\neq i}X^i$ be the sample mean of the random variables excluding the $i$-th one, and
\[
\bar{S}^{-i} = \sqrt{\frac{1}{n}\sum_{k\neq i}(X_k - \bar{X}^{-i})^2}.
\]
Then, it follows from Lemma~\ref{lem:1} that
\[
|Z^i| = \frac{|X^i-\bar{X}^{n+1}|}{\bar{S}^{n+1}}
    = \frac{|X^i - \frac{1}{n+1}\sum^{n+1}_{i=1}X^i|}{\sqrt{\frac{1}{n}\sum^{n+1}_{i=1}(X^i - \bar{X}^{n+1})^2}}=\frac{\frac{n}{n+1}|X^i - \bar{X}^{-i}|}{\sqrt{\frac{1}{n+1}(X^i-\bar{X}^{-i})^2+\bar{S}^{-i}}}.
\]
But $\bar{S}^{-i} \geq 0$, so we must have
\begin{align*}
    |Z^i|\leq \frac{\frac{n}{n+1}|X^i - \bar{X}^{-i}|}{\sqrt{\frac{1}{n+1}(X^i-\bar{X}^{-i})^2}}= \frac{n}{\sqrt{n+1}}.
\end{align*}
If $X^1,\ldots,X^{n+1}$ are almost-surely distinct, we know $\bar{S}^{-i} >0$ almost-surely for any $i \in [n+1]$ thus this inequality becomes strict.
\end{proof}

\subsection{Proofs of Main Results}
\label{app:proofs-thms}

\begin{proof}[of Theorem~\ref{thm:fixed-transformation-coverage}]
This is a well-known result, whose proof is included for completeness. 
By definition of $\hat{C}(X^{n+1})$ in~\eqref{eq:joint-prediction-scalar},
\[
 Y^{n+1} \notin \hat{C}(X^{n+1}) \Longleftrightarrow S^{n+1} > \hat{Q}_{1-\alpha}(S^1,\ldots,S^n).
\]
Moreover, $S^1,\ldots,S^{n+1}$ are exchangeable because $E^1,\ldots,E^{n+1}$ are exchangeable and $\Phi$ is fixed.
Therefore, the proof is complete by applying the fundamental ``quantile inflation lemma'' of conformal prediction (e.g., see \citet{angelopoulos2025theoreticalfoundationsconformalprediction}), which states 
\begin{align*}
  \mathbb{P}\left[ S^{n+1} > \hat{Q}_{1-\alpha}(S^1,\ldots,S^n) \right] \leq \alpha.
\end{align*}
Moreover, it is also a standard result that, if $S^1,\ldots,S^{n+1}$ are almost-surely unique,
\begin{align*}
  \mathbb{P}\left[ S^{n+1} > \hat{Q}_{1-\alpha}(S^1,\ldots,S^n) \right] \geq \alpha - \frac{1}{n+1}.
\end{align*}
\end{proof}

\begin{proof}[of Lemma~\ref{lem:solution-key-ineq}]
From Assumption~\ref{eq:assumption-scores}, we know that $E^1_j,\dots, E^{n+1}_j$ are almost-surely distinct for every $j \in [d]$ and $0<\hat{\sigma}^{\mathrm{oracle}}_j < \infty$ almost-surely for all $j \in [d]$, which also ensures $S_{\mathrm{oracle}}^{n+1}$ is well-defined. 

To prove this lemma, fix any $c \in \R$ and note that
\begin{align*}
    S_{\mathrm{oracle}}^{n+1} \leq c &\iff E^{n+1}_j \leq c\hat{\sigma}_j^{\mathrm{oracle}} + \hat{\mu}_j^{\mathrm{oracle}}, \quad \forall j \in [d]\\
    &\iff E^{n+1}_j \leq c\sqrt{\frac{1}{n}\sum^{n+1}_{i=1}\left(E^i_j -\hat{\mu}_j^{\mathrm{oracle}} \right)^2} + \frac{1}{n+1}\sum^{n+1}_{i=1}E^i_j, \quad \forall j \in [d].
\end{align*}
Everything except $E^{n+1}$ here is observed, so we can solve this inequality for $E^{n+1}$ coordinate-wise.
Fix any $j \in [d]$ and define
\begin{align*}
    g_j(E^{n+1})&:= E^{n+1}_j - c\sqrt{\frac{1}{n}\sum^{n+1}_{i=1}\left(E^i_j -\hat{\mu}_j^{\mathrm{oracle}} \right)^2} - \frac{1}{n+1}\sum^{n+1}_{i=1}E^i_j\\
    &= \frac{n}{n+1}\left(E^{n+1}_j-\hat{\mu}_j\right) - c\sqrt{\frac{1}{n}\sum^{n+1}_{i=1}\left(E^i_j -\hat{\mu}_j^{\mathrm{oracle}} \right)^2}\\
    &= \frac{n}{n+1}\left(E^{n+1}_j-\hat{\mu}_j\right) - c\sqrt{\frac{1}{n+1}\left(E^{n+1}_j-\hat{\mu}_j\right)^2+\hat{\sigma}_j^2}.
\end{align*}
where the last equality follows from Lemma~\ref{lem:1}. Therefore, 
\begin{align}
    g_j(E^{n+1}) \leq 0 \iff \frac{n}{n+1}\left(E^{n+1}_j-\hat{\mu}_j\right) \leq c\sqrt{\frac{1}{n+1}\left(E^{n+1}_j-\hat{\mu}_j\right)^2+\hat{\sigma}_j^2}.
    \label{eq:coordinte-wise-key-ineq}
\end{align}
The solution of this inequality in $E^{n+1}_j$ depends on the range of $c$. There are two cases:
\begin{enumerate}
\item if $c \geq 0$, then~\eqref{eq:coordinte-wise-key-ineq} holds if and only if \emph{either} $0\leq E^{n+1}_j < \hat\mu_j$ \emph{or} $E^{n+1}_j \geq \mu_j$ and
    \[
    \left(\frac{n}{n+1}\right)^2\left(E^{n+1}_j-\hat\mu_j\right)^2 \leq c^2\left(\frac{1}{n+1}(E^{n+1}_j-\hat\mu_j)^2 + \hat\sigma_j^2\right).
    \]
\item if $c < 0$, then~\eqref{eq:coordinte-wise-key-ineq} holds if and only if $E^{n+1}_j <  \hat\mu_j$ \emph{and}
    \begin{align*}
        \left(\frac{n}{n+1}\right)^2\left(E^{n+1}_j-\hat\mu_j\right)^2 \geq c^2\left(\tfrac{1}{n+1}\left(E^{n+1}_j-\hat\mu_j\right)^2 + \hat\sigma_j^2\right).
    \end{align*}
  \end{enumerate}
  The solution in both cases relies on solving a quadratic equality in the form $h(x) = 0$, where:
\[
h(x):=\left(\frac{n}{n+1}\right)^2\left(x-\hat\mu_j\right)^2 - c^2\tfrac{1}{n+1}\left(x-\hat\mu_j\right)^2 - c^2\hat\sigma_j^2 = 0.
\]
This equality can be reorganized into a more recognized quadratic form 
\[
h(x)= Ax^2+Bx +K = 0,
\]
where
\begin{align*}
 A & = \left(\frac{n}{n+1}\right)^2 - \frac{c^2}{n+1}, \\
 B & = 2\hat\mu_j\left(\frac{n}{n+1}\right)^2-2\hat\mu_j \frac{c^2}{n+1}, \\
 K & = \left(\frac{n}{n+1}\right)^2\hat\mu_j^2-\frac{c^2}{n+1}\hat\mu_j^2-c^2\hat\sigma_j^2.
\end{align*}
The solution set is:
\[h(x) = 0 \iff x=
\begin{cases}
    \hat\mu_j \pm \hat\sigma_j|c|\sqrt{\tfrac{(n+1)^2}{n^2-(n+1)c^2}}, & c^2 < \tfrac{n^2}{n+1},\\
    \hat{\mu}_j, & c^2 = \tfrac{n^2}{n+1} \text{ and } \hat{\sigma}_j = 0,\\
    \text{no real solution}, &\text{otherwise}.
\end{cases}
\]
We first note the event $\hat{\sigma}_j = 0$ occurs if and only if $E^1_j=\ldots=E^n_j$, which has probability zero by Assumption~\ref{eq:assumption-scores}. For the rest, we note that $c^2 < \tfrac{n^2}{n+1}$ implies $A>0$ thus $h(x) = Ax^2+Bx +K$ concave up for all $x \geq 0$. 
Therefore by analyzing the positiveness and negativeness of $h(x)$ relative to the solution set in our two cases, we know~\eqref{eq:coordinte-wise-key-ineq} holds almost-surely if and only if
\begin{align*}
    &\begin{cases}
    \hat\mu_j \leq E^{n+1}_j \leq \hat\mu_j+\hat\sigma_j|c|\sqrt{\tfrac{(n+1)^2}{n^2-(n+1)c^2}}, & c \geq 0,\, c^2 < \tfrac{n^2}{n+1}, \\[0.5em]
    E^{n+1}_j < \infty , & c \geq 0,\, c^2 \geq \tfrac{n^2}{n+1}, \\[0.5em]
    0 \leq E^{n+1}_j < \hat\mu_j, & c \geq 0, \\[0.5em]
    0\leq E^{n+1}_j \leq \hat\mu_j-\hat\sigma_j|c|\sqrt{\tfrac{(n+1)^2}{n^2-(n+1)c^2}}, & c < 0,\, c^2 < \tfrac{n^2}{n+1}, \\[0.5em]
    \text{no solution}, & c < 0,\, c^2 \geq \tfrac{n^2}{n+1}.
    \end{cases} \quad \text{almost-surely}\\
    &\iff \begin{cases}
    0 \leq E^{n+1}_j \leq \hat\mu_j+\hat\sigma_j|c|\sqrt{\tfrac{(n+1)^2}{n^2-(n+1)c^2}}, & 0\leq c < \tfrac{n}{\sqrt{n+1}}, \\[0.5em]
    0\leq E^{n+1}_j < \infty , & c>\tfrac{n}{\sqrt{n+1}}, \\[0.5em]
    0\leq E^{n+1}_j \leq \max\left\{0, \hat\mu_j-\hat\sigma_j|c|\sqrt{\tfrac{(n+1)^2}{n^2-(n+1)c^2}}\right\}, & -\tfrac{n}{\sqrt{n+1}} < c<0, \\[0.5em]
    \text{no solution}, & c \leq - \tfrac{n}{\sqrt{n+1}}.
    \end{cases}\quad \text{almost surely}
\end{align*}
Interpreting the ``no-solution'' case as $E^{n+1}_j = 0$, we obtain precisely the expression for $\omega_j(c)$ given in the statement of the lemma. This notation adjustment is harmless because both $E^{n+1}_j = 0$ and $S^{n+1}_{\mathrm{oracle}} \leq -n/\sqrt{n+1}$ are probability null sets by our Assumption~\ref{eq:assumption-scores} and Lemma~\ref{lem:3}, respectively.
The monotonic nondecreasing property of this $\omega_j(c)$ follows immediately from the construction.
\end{proof}

\begin{proof}[of Lemma~\ref{lem:finite-prediction-guarantee}]
From Assumption~\ref{eq:assumption-scores}, $E^1_j,\dots, E^{n+1}_j$ are almost-surely distinct for every $j \in [d]$ and $0<\hat{\sigma}^{\mathrm{oracle}}_j < \infty$ almost-surely for all $j \in [d]$. This ensures $S_{\mathrm{oracle}}^i$ is well-defined for all $i \in [n+1]$. Moreover, we know from Lemma~\ref{lem:3} that
\[
\left|\frac{E^i_j - \hat{\mu}^{\mathrm{oracle}}_j}{\hat{\sigma}^{\mathrm{oracle}}_j}\right| < \frac{n}{\sqrt{n+1}},
\]
almost-surely for all $i \in [n+1]$. Aggregating this inequality for all $j \in [d]$, we obtain
\begin{align*}
    |S_{\mathrm{oracle}}^i| = \left|\max_{1\leq j\leq d}\frac{E^i_j - \hat{\mu}^{\mathrm{oracle}}_j}{\hat{\sigma}^{\mathrm{oracle}}_j}\right|\leq \max_{1\leq j\leq d}\left|\frac{E^i_j - \hat{\mu}^{\mathrm{oracle}}_j}{\hat{\sigma}^{\mathrm{oracle}}_j}\right|< \max_{1\leq j\leq d}\frac{n}{\sqrt{n+1}}=\frac{n}{\sqrt{n+1}},
\end{align*}
almost-surely for all $i \in [n+1]$. Recall that the quantile $\hat{Q}^{\mathrm{oracle}}_{1-\alpha}$ is defined as
\[
\hat{Q}^{\mathrm{oracle}}_{1-\alpha}:= \lceil(1-\alpha)(n+1)\rceil\text{th smallest value in } \{S_{\mathrm{oracle}}^1,\dots, S_{\mathrm{oracle}}^n, +\infty\},
\]
so $\hat{Q}^{\mathrm{oracle}}_{1-\alpha} < n/\sqrt{n+1}$ almost-surely if $n \geq 1/\alpha-1$ and $\hat{Q}^{\mathrm{oracle}}_{1-\alpha} = +\infty$ otherwise.
\end{proof}


\begin{proof}[of Lemma~\ref{lem:gwc-rescaling}]
Fix any $j\in [d]$ and let
\[
g(t_j, z_j) := \frac{t_j - \hat{\mu}_j(z_j)}{\hat{\sigma}_j(z_j)}.
\]
Finding the supremum of this $g$ over $z_j \geq 0$ is a simple univariate optimization problem. Expanding $g$ in full, we get
\[
g(t_j, z_j) = \frac{t_j - \hat{\mu}_j(z_j)}{\hat{\sigma}_j(z_j)} = \frac{t_j-\frac{1}{n+1}(z_j+n\hat{\mu}_j)}{\sqrt{\frac{1}{n+1}(z_j-\hat{\mu}_j)^2+\hat{\sigma}_j^2}},
\]
where the last equality follows from Lemma~\ref{lem:1}.
Differentiating $g$ with respect to $z_j$ yields
\[
g_{z_j}(t_j, z_j) = 
\frac{-\tfrac{1}{n+1}\sqrt{\frac{1}{n+1}(z_j-\hat{\mu}_j)^2+\hat{\sigma}_j^2} 
- \big(t_j - \tfrac{1}{n+1}(z_j+n\hat\mu_j)\big)\tfrac{\tfrac{1}{n+1}(z_j-\hat\mu_j)}{\sqrt{\frac{1}{n+1}(z_j-\hat{\mu}_j)^2+\hat{\sigma}_j^2}}}
{\frac{1}{n+1}(z_j-\hat{\mu}_j)^2+\hat{\sigma}_j^22}.
\]
Solving $g_{z_j}(t_j,z^*_j)=0$ with respect to $z^*_j$ gives:
\[
\hat\sigma_j^2 + (t_j-\hat\mu_j)(z^*_j-\hat\mu_j)=0
\quad\Longrightarrow\quad
z^\ast_j = \hat\mu_j - \frac{\hat\sigma_j^2}{t_j-\hat\mu_j}.
\]
There are three cases here for any fixed $t_j\geq 0$:
\begin{itemize}
    \item Case 1: if $g(t_j, z^\ast_j)$ is indeed the global maximum and $z^\ast_j > 0$, then we can just take
    \[
    \sup_{z_j \geq 0} g(t_j, z_j) = g(t_j, z^\ast_j).
    \]
    \item Case 2: if $g(t_j, z^\ast_j)$ is not the global maximum or $z^\ast_j < 0$, then the supremum must occur at either $z_j = 0$ or $z_j \rightarrow \infty$ due to the fact that $g_{z_j}$ has only one critical point. Hence,
    \[
    \sup_{z_j \geq 0} g(t_j, z_j) = \max\left\{g(t_j, 0), \lim_{z_j\to\infty} g(t_j,z_j)\right\} = \max\left\{g(t_j, 0), -\frac{1}{\sqrt{n+1}}\right\} .
    \]
    \item Case 3: if $t_j = \hat{\mu}_j$, then similar derivation strategy shows that $g_{z_j} \equiv -\hat{\sigma}_j^2 < 0$. This implies that $g$ is monotonically decreasing and thus we take
    \[
    \sup_{z_j \geq 0} g(\hat{\mu}_j, z_j) = g(\hat{\mu}_j, 0).
    \]
\end{itemize}

Collecting all candidates, we obtain the closed form:
\[
\sup_{z_j \geq 0} g(t_j, z_j)
=\max\left\{g\left(t_j, 0\right), g\left(t_j, z^*_j\mathbb{I}\left\{z^*_j \geq 0\right\}\right), -\frac{1}{\sqrt{n+1}}\right\}.
\]
\end{proof}

\begin{proof}[of Proposition~\ref{prop:lwc-approximation}]
Let $\mathcal{P}$ be any (random) partition of $\hat{C}^{\mathrm{gwc}}(X^{n+1})$, so that $\hat{C}^{\mathrm{gwc}}(X^{n+1}) = \bigsqcup_{\mathcal{R} \in \mathcal{P}}\mathcal{R}$ almost-surely. Because $\hat{C}^{\mathrm{gwc}}(X^{n+1}) \supseteq \tilde{C}(X^{n+1})$ almost-surely, we know
\begin{align*}
    Y^{n+1} \in \tilde{C}(X^{n+1})
    &\iff Y^{n+1} \in \tilde{C}(X^{n+1}), Y^{n+1} \in \hat{C}^{\mathrm{gwc}}(X^{n+1})\\
    &\iff \exists \mathcal{R} \in \mathcal{P}:  Y^{n+1} \in \tilde{C}(X^{n+1}), Y^{n+1} \in \mathcal{R}\\
    &\iff \exists \mathcal{R} \in \mathcal{P}:  Y^{n+1} \in \tilde{C}(X^{n+1}) \cap \mathrm{R}.
\end{align*}
If $\hat{C}^\mathrm{lwc}(X^{n+1}, \mathcal{R})$ satisfies the condition on the right-hand-side of~\eqref{eq:lwc-local-containment},
\[
\mathcal{R}\supseteq \hat{C}^\mathrm{lwc}(X^{n+1}, \mathcal{R}) \supseteq \tilde{C}(X^{n+1}) \cap \mathrm{R},
\]
then
\begin{align*}
    Y^{n+1} \in \tilde{C}(X^{n+1}) &\implies \exists \mathcal{R}\in \mathcal{P}, Y^{n+1} \in \hat{C}^\mathrm{lwc}(X^{n+1}, \mathcal{R})\\
 &\iff Y^{n+1} \in \hat{C}^\mathrm{lwc}(X^{n+1}) = \bigsqcup_{\mathcal{R} \in \mathcal{P}}\hat{C}^\mathrm{lwc}(X^{n+1}, \mathcal{R}).
\end{align*}
Since we already know that $\mathbb{P}[Y^{n+1} \in \tilde{C}(X^{n+1})] \geq 1-\alpha$, this proves that
\begin{align*}
  \mathbb{P} \left[ Y^{n+1} \in \hat{C}^\mathrm{lwc}(X^{n+1}) \right] \geq 1-\alpha.
\end{align*}
Finally, the containment that $\hat{C}^\mathrm{lwc}(X^{n+1}) \subseteq \hat{C}^\mathrm{gwc}(X^{n+1})$ almost-surely is trivial because 
\[
\hat{C}^\mathrm{lwc}(X^{n+1}, \mathcal{R}) \subseteq \mathcal{R} \subseteq \hat{C}^\mathrm{gwc}(X^{n+1}),
\]
for every $\mathcal{R} \in \mathcal{P}$.
\end{proof}

\begin{proof}[of Lemma~\ref{lem:lwc-rescaling}]
Fix any $h \in [n+1]^d$.
Suppose $\mathcal{Y}^h \neq \emptyset$, then $R^h \neq \emptyset$ and
by Equation~\eqref{eq:rectangle-explicit},
\[
\mathcal{I}_j^h :=\left[E^{(h_j-1)}, U^{h_j}_j\right] \neq \emptyset, \quad \forall j \in [d].
\]
Fix any $t \in \R^d_+$. Similar to the proof of Lemma~\ref{lem:gwc-rescaling}, we have to solve a univariate optimization problem. We first notice that for any $j \in [d]$,
\[
\sup_{z \in R^h} \frac{t_j}{\hat{\sigma}_j(z_j)}= \sup_{z_j \in \mathcal{I}_j^h} \frac{t_j}{\hat{\sigma}_j(z_j)} = \frac{t_j}{\inf_{z_j \in \mathcal{I}_j^h}\hat{\sigma}_j(z_j)}.
\]
By Lemma~\ref{lem:1}, we know that
\[
\hat{\sigma}_j(z_j) = \sqrt{\frac{1}{n+1}(z_j - \hat{\mu}_j)^2 + \hat{\sigma}_j^2},
\]
so clearly $\hat{\sigma}_j(z_j)$ is minimized at 
\[
z^*_j = \operatorname*{argmin}_{z_j \in \mathcal{I}_j^h}|z_j-\hat{\mu}_j| = \begin{cases}
    \hat{\mu}_j, &\text{ if }\hat{\mu}_j \in \mathcal{I}_j^h,\\
    E^{(h_j-1)}, &\text{ if } \hat{\mu}_j< E^{(h_j-1)},\\
    U^{h_j}_j, &\text{ if }\hat{\mu}_j >  U^{h_j}_j.
\end{cases}
\]
Hence, $r_j(h_j) = \hat{\sigma}_j(z^*_j)$ is exactly the infimum proposed. 
\noindent
For the second part, denote
\[
g(z_j):=\frac{\hat{\mu}_j(z_j)}{\hat{\sigma}_j(z_j)} 
= \frac{\tfrac{1}{n+1}(z_j+n\hat\mu_j)}{\sqrt{\tfrac{1}{n+1}(z_j-\hat\mu_j)^2+\hat\sigma_j^2}}=\frac{\tfrac{1}{n+1}(z_j-\hat\mu_j)+\hat{\mu}_j}{\sqrt{\tfrac{1}{n+1}(z_j-\hat\mu_j)^2+\hat\sigma_j^2}},
\]
then we want to find
\[
\inf_{z \in \mathcal{E}^{\mathrm{gwc}}}g(z_j) = \inf_{z_j \in [0, \omega_j(\hat{Q}^{\mathrm{gwc}}_{1-\alpha})]}g(z_j).
\]
Differentiating $g(z_j)$ in $z_j$, we get
\[
g'(z_j) = \frac{\frac{1}{n+1}\hat\sigma_j^2-\frac{1}{n+1}(z_j-\hat\mu_j)\hat\mu_j}{(\tfrac{1}{n+1}(z_j-\hat\mu_j)^2+\hat\sigma_j^2)\sqrt{\tfrac{1}{n+1}(z_j-\hat\mu_j)^2+\hat\sigma_j^2}}.
\]
Set this to zero gives the critical point
\[
z^\ast_j = \mu_j + \frac{\sigma_j^2}{\mu_j}> 0.
\]
Let $\epsilon>0$, then we see that
\begin{align*}
    \mathrm{sign}(g'(z^*_j+\epsilon)) &= \mathrm{sign}(-\frac{1}{n+1}\epsilon\hat{\mu}_j) = -1,\\
    \mathrm{sign}(g'(z^*_j-\epsilon)) &= \mathrm{sign}(\frac{1}{n+1}\epsilon\hat{\mu}_j) = +1.
\end{align*}
This implies that $g(z^*)$ is always the local maximum. Hence, if $z^*_j \in [0, \omega_j(\hat{Q}^{\mathrm{gwc}}_{1-\alpha})]$, then 
\[
\inf_{z_j \in [0, \omega_j(\hat{Q}^{\mathrm{gwc}}_{1-\alpha})]}g(z_j) = \min\left\{g(0), \lim_{z_j \rightarrow \omega_j(\hat{Q}^{\mathrm{gwc}}_{1-\alpha})}g(z_j)\right\};
\]
if $z^*_j > \omega_j(\hat{Q}^{\mathrm{gwc}}_{1-\alpha})$, then 
\[
\inf_{z_j \in [0, \omega_j(\hat{Q}^{\mathrm{gwc}}_{1-\alpha})]}g(z_j) = g(0).
\]
But in either case, we only need to compare the boundaries in order to find the infimum. Therefore, it follows that
\[
\inf_{z_j \in [0, \omega_j(\hat{Q}^{\mathrm{gwc}}_{1-\alpha})]}g(z_j) = \min\left\{g(0), \lim_{z_j \rightarrow \omega_j(\hat{Q}^{\mathrm{gwc}}_{1-\alpha})}g(z_j)\right\}.
\]
This is true for every $j \in [d]$, so we must have
\begin{align*}
  \hat{\Phi}^{\mathrm{lwc}}(t; h) 
  &=\max_{1\leq j\leq d}\left\{\sup_{z \in R^h} \frac{t_j}{\hat{\sigma}_j(z_j)}-\inf_{z \in \mathcal{E}^{\mathrm{gwc}}}g(z_j)\right\}\\
  & = \max_{1 \le j \le d}
\left\{\frac{t_j}{m_j}-\min\left\{\frac{\hat{\mu}_j(0)}{\hat{\sigma}_j(0)}, \lim_{z_j \rightarrow \omega_j(\hat{Q}^{\mathrm{gwc}}_{1-\alpha})}\frac{\hat{\mu}_j(z_j)}{\hat{\sigma}_j(z_j)}\right\}\right\}.
\end{align*}
\end{proof}

\begin{proof}[of Lemma~\ref{lem:reduction-search}]
Fix any $j\in[d]$ and suppose $\mathcal{Y}^{h^*}\neq\emptyset$. There are two parts to this lemma, which we prove separately.

\textit{Part I.} We aim to prove
\begin{equation}
\sup_{h\in[n+1]^d}B^h_j
=\sup_{h\in\mathrm{Row}_j(h^*)}B^h_j
=B^{h^{[j]}}_j.
\label{eq:reduction-to-show}
\end{equation}
We first notice that the first supremum can be decomposed into suprema over surfaces:
\[
\sup_{h \in [n+1]^d}B^h_j = \sup_{k \in [n+1]}\sup_{h \in \mathrm{Surface}_j(k)}B^h_j,
\]
where
\[
  \mathrm{Surface}_j(k) := \left\{h \in [n+1]^d: h_j = k \right\}. 
\]
Fix any $k \in [n+1]$ and let $h^*({j,k}) \in \mathrm{Surface}_j(k)$ be the multi-index defined such that $h_j^*({j,k}) = k$ and $h_l^*({j,k}) = h^*_l$ for all $l \neq j$; then $h^*(j,k) \in \mathrm{Row}_j(h^*)$. Since this holds for every $k \in [n+1]$, to prove~\eqref{eq:reduction-to-show} the main challenge is to show that:
\begin{equation}
    \sup_{h \in \mathrm{Surface}_j(k)}B^h_j = B^{h^*({j,k})}_j.
    \label{eq:surface}
\end{equation}

To prove~\eqref{eq:surface}, recall from the definition in~\eqref{eq:lwc-residual-transformation} that
\[
\hat\Phi^{\mathrm{lwc}}(t,h)  = \max_{1\leq l \leq d} \left\{\frac{t_l}{r_l(h_l)} - c_l\right\},
\]
where
\[
r_l(h_l) = \min\left\{\hat{\sigma}_l(z_l): z_l \in R_l^h: =\left[E^{(h_l-1)}_l, U^{h_l}_l\right]\right\},
\qquad c_l = \inf_{z \in \mathcal{E}^{\mathrm{gwc}}}\frac{\hat{\mu}_l(z_l)}{\hat{\sigma}_l(z_l)}.
\]
Without loss of generality, assume $\mathcal{Y}^h \neq \emptyset$ for all $h \in \mathrm{Surface}_j(k)$. This does not affect the final result because $\mathcal{Y}^h = \emptyset$ leads to $B^h_j = 0$, which can be safely excluded from the supremum calculation. Then notice that, for all $h \in \mathrm{Surface}_j(k)$, we have $r_j(h_j) = r_j(k)$ and $r_l(h_l) \geq \hat{\sigma}_l = r_l({h^*(j,k)})$. Therefore, for any fixed $t$,
\begin{align*}
    \hat\Phi^{\mathrm{lwc}}(t,h) & = \max_{1\leq l \leq d} \left\{\frac{t_l}{r_l(h_l)} - c_l\right\}
    \leq  \max\left(\max_{l\neq j} \left\{\frac{t_l}{\hat{\sigma}_l} - c_l\right\}, \frac{t_j}{r_j(k)} - c_j\right)
    = \Phi^{\mathrm{lwc}}(t, h^*(j,k)).
\end{align*}
Because this holds point-wise for all $t$, we can extend the inequality to the corresponding conformal quantile and conclude that, almost surely for all $h \in \mathrm{Surface}_j(k)$,
\[
\hat{Q}^\mathrm{lwc}_{1-\alpha}(h) \leq \hat{Q}^\mathrm{lwc}_{1-\alpha}(h^*(j,k))
\]
and therefore, by the monotonicity of the function $\omega$ from Lemma~\ref{lem:solution-key-ineq},
\[
\omega_j(\hat{Q}^\mathrm{lwc}_{1-\alpha}(h)) \leq \omega_j(\hat{Q}^\mathrm{lwc}_{1-\alpha}(h^*(j,k))).
\]
Therefore, it follows that, for all $h \in \mathrm{Surface}_j(k)$,
\begin{align*}
B^h_j &=\begin{cases}
\min\left\{U^{k}_j, \omega_j(\hat{Q}^{\mathrm{lwc}}_{1-\alpha}(h))\right\}, 
& \text{if } U^{k}_j,\omega_j(\hat{Q}^{\mathrm{lwc}}_{1-\alpha}(h) > E^{(k-1)}_j\\[0.4em]
0, & \text{otherwise}
 \end{cases}\\
&\leq \begin{cases}
\min\!\bigl\{U^{k}_j, \omega_j(\hat{Q}^{\mathrm{lwc}}_{1-\alpha}(h^*(j,k)))\bigr\}, 
& \text{if } U^{k}_j,\omega_j(\hat{Q}^{\mathrm{lwc}}_{1-\alpha}(h^*(j,k))) > E^{(k-1)}_j\\[0.4em]
0, & \text{otherwise}
\end{cases}\\
&=B^{h^*(j,k)}_j.
\end{align*}
This completes the proof of~\eqref{eq:surface}. 

The second equality in~\eqref{eq:reduction-to-show}
is trivial. We can easily see that for any $k \in [n]$, if $B^{h^*(j,k+1)}_j \neq 0$, then
\[
U^{k+1}_j, \omega_j(\hat{Q}^{\mathrm{lwc}}_{1-\alpha}(h^*(j,k+1)) \geq E^{(k)}_j \geq \min\left\{E^{(k)}_j, \omega_j(\hat{Q}^{\mathrm{gwc}}_{1-\alpha})\right\} = U^k_j.
\]
Hence, it follows that $B^{h^*(j,k+1)}_j \geq B^{h^*(j,k)}_j$ regardless of whether $B^{h^*(j,k)}_j = 0$ or not. Then finding the supremum among all $k\in [n+1]$ simplifies to finding the right-most nonzero $B_j^{h^*(j,k)}$, which is exactly how $B^{h^{[j]}}_j$ is defined.

\textit{Part II.} In this part, we aim to show that if $B_j^{h^*(j,m_j)} =0$ for some $m_j \geq h^*_j$, then
\[
B_j^{h^*(j,N)} = 0, \quad \forall N\geq m_j.
\]
There two cases that could lead to $B_j^{h^*(j,m_j)} =0$:
    \begin{itemize}
        \item If $U^{m_j}_j = \min\left\{E^{(m_j)}_j,\omega_j(\hat{Q}^{\mathrm{gwc}}_{1-\alpha})\right\} \leq E^{(m_j-1)}_j$, then $\omega_j(\hat{Q}^\mathrm{gwc}_{1-\alpha}) \leq E^{(m_j-1)}_j \leq E^{(N-1)}_j$ for any $N\geq m_j$ by the nature of order statistics. Hence, $U^N_j \leq E^{(N-1)}_j$ for all $N\geq m_j$, which further implies $B_j^{h^*(j,N)} = 0$ for all $N\geq m_j$.
        \item If $\omega_j(\hat{Q}^\mathrm{lwc}_{1-\alpha} (h^*(j,m_j)) \leq E^{(m_j-1)}_j$, then we claim that
        \[
        \omega_j(\hat{Q}^\mathrm{lwc}_{1-\alpha} (h^*(j,N)) \leq \omega_j(\hat{Q}^\mathrm{lwc}_{1-\alpha} (h^*(j,m_j)) \leq E^{(m_j-1)}_j \leq E^{(N-1)}_j,
        \]
        which then implies $B_j^{h^*(j,N)} =0$ for all $N\geq m_j$.
        This is actually trivial because
        \[
        r_l(h^*(j, m_j)) = \min_{z_j \in R_j^{h^*(j, m_j)}} |z_j - \hat{\mu}_j| \leq \min_{z_j \in R_j^{h^*(j,N)}} |z_j - \hat{\mu}_j| = r_l(h^*(j, N)),
        \]
        according to how $R_j^h$ is defined earlier. This then leads to
        \begin{align*}
            &\Phi^{\mathrm{lwc}}(t, h^*(j, m_j)) \geq \Phi^{\mathrm{lwc}}(t, h^*(j,N)), \quad \forall t\\
            &\implies \hat{Q}^\mathrm{lwc}_{1-\alpha} (h^*(j,N)) \leq \hat{Q}^\mathrm{lwc}_{1-\alpha}(h^*(j, m_j))\\
            &\implies \omega_j(\hat{Q}^\mathrm{lwc}_{1-\alpha} (h^*(j,N))) \leq \omega_j(\hat{Q}^\mathrm{lwc}_{1-\alpha} (h^*(j,m_j))).
        \end{align*}
    \end{itemize}
\end{proof}

\clearpage

\section{Additional Numerical Results}
\label{app:additional-experiments}

This section presents the results of additional numerical experiments, including comparisons with a broader range of alternative approaches.
These experiments are conducted using the same protocol and performance metrics described in Section~\ref{sec:experiments} unless specified otherwise. 

\subsection{Additional Results from the Experiments of Section~\ref{sec:experiments}}
\label{app:simulations}

Here we present additional results from the numerical experiments described in Section~\ref{sec:experiments}, corresponding to six different noise distributions. These results are qualitatively consistent with those summarized in Section~\ref{sec:experiments}: our method performs consistently well, producing relatively tight prediction sets with valid coverage.

\subsubsection{Gaussian Noise (heterogeneous)}

\begin{table}[!htbp]
    \centering
    \input{figures/baselines/gaussian_10d_table}
    \caption{Performance of the proposed TSCP method and alternative conformal prediction approaches on simulated data with $d=10$ outcome variables, as in Figure~\ref{fig:gaussian-heter}. All results are averaged over $200$ trials with their one standard deviation reported in parenthesis.  
    The target coverage level is $0.9$. Coverage below this target is highlighted in \textcolor{red}{red}. For volume, lower is better. Smallest volume size (among those with valid theoretical coverage) is highlighted in \textbf{bold}. TSCP yields the tightest volume while maintaining valid coverage.}
    \label{tab:gaussian-heter}
\end{table}

\clearpage
\subsubsection{Gaussian Noise (homogeneous)}

\begin{table}[!htbp]
    \centering
    \input{figures/baselines/unit_gaussian_10d_table}
    \caption{Performance of the proposed TSCP method and alternative conformal prediction approaches on simulated data with $d=10$ outcome variables, as in Figure~\ref{fig:gaussian-homo}. Other details are as in Table~\ref{tab:gaussian-heter}. In this case, Unscaled Max yields the tightest volume while maintaining valid coverage, while TSCP remains a competitive second. }
    \label{tab:gaussian-homo}
\end{table}

\clearpage
\subsubsection{Laplace Noise}
\begin{table}[!htbp]
    \centering
    \input{figures/baselines/laplace_10d_table}
    \caption{Performance of the proposed TSCP method and alternative conformal prediction approaches on simulated data with $d=10$ outcome variables, and noise following a \emph{Laplace} distribution; i.e., $\epsilon _j \sim \mathrm{Laplace}(10-j+1)$ for all $j \in [10]$. Other details are as in Table~\ref{tab:gaussian-heter}.}
    \label{tab:laplace}
\end{table}

\clearpage
\subsubsection{Mixture Noise}

\begin{table}[!htbp]
    \centering
    \input{figures/baselines/mixed_10d_table}
    \caption{Performance of the proposed TSCP method and alternative conformal prediction approaches on simulated data with $d=10$ outcome variables, and noise following a \emph{Mixture} distribution; i.e., $\epsilon_j \sim \frac{1}{2}\mathrm{Laplace}(10-j+1)+\frac{1}{2}\mathcal{N}(0, (10-j+1)^2$ for all $j \in [10]$. Other details are as in Table~\ref{tab:gaussian-heter}.}
    \label{tab:mixed}
\end{table}

\clearpage
\subsubsection{Gamma Noise}
 
\begin{table}[!htbp]
    \centering
    \input{figures/baselines/gamma_10d_table}
    \caption{Performance of the proposed TSCP method and alternative conformal prediction approaches on simulated data with $d=10$ outcome variables, and noise following a \emph{Gamma} distribution; i.e., $\epsilon_j \sim  \mathrm{Gamma}(1, 10-j+1)$ for all $j \in [10]$. Other details are as in Table~\ref{tab:gaussian-heter}.}
  \end{table}

\clearpage
  \subsection{Comparison of Alternative Implementations of TSCP}
\label{app:ours_simulation}

The additional methods considered in this section include (i) \emph{Naive}, the heuristic plug-in rescaling method that uses location parameters estimated directly from calibration data without further adjustments, summarized in Algorithm~\ref{alg:std-plug-in}; (ii) \emph{TSCP-S}, the inefficient data-splitting variant summarized in Algorithm~\ref{alg:std-ds}; (iii) \emph{TSCP-GWC}, our conservative global-worst-case method summarized in Algorithm~\ref{alg:global}; (iv) \emph{TSCP-LWC}, the aggregated union described in Section~\ref{sec:method-lwc} but not enclosed by a rectangle.

\begin{figure}[!htbp]
    \begin{subfigure}[b]{\textwidth}
        \centering
        \includegraphics[width=0.9\textwidth]{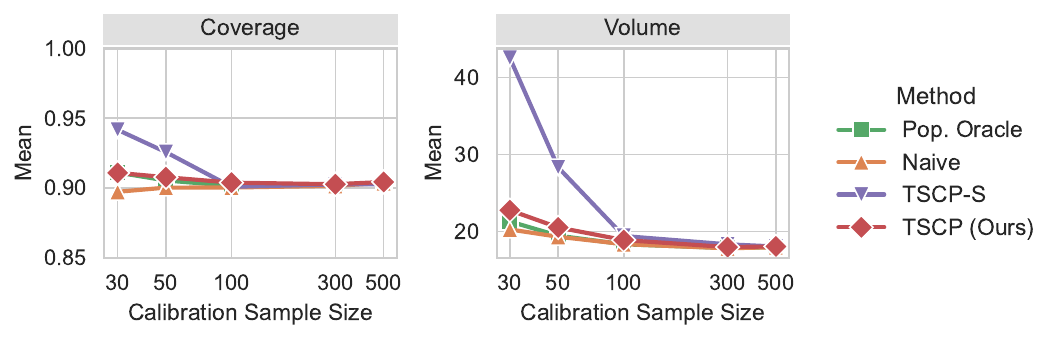}
    \end{subfigure}
    \begin{subfigure}[b]{\textwidth}
        \centering
        \includegraphics[width=0.9\textwidth]{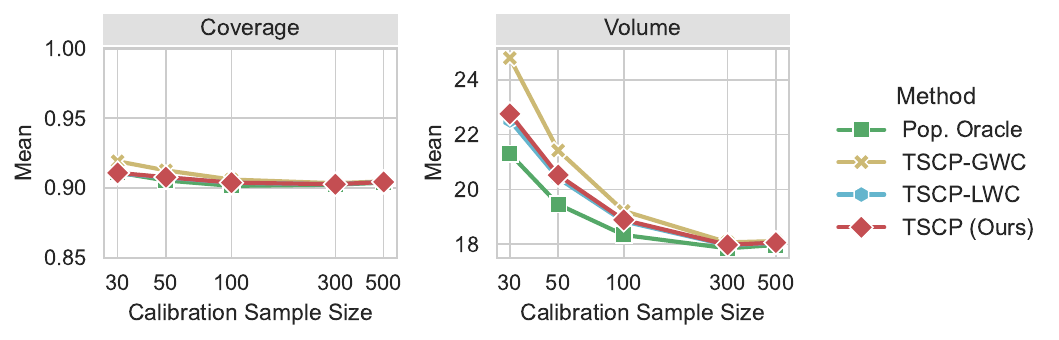}
    \end{subfigure}
    \caption{Performance of TSCP and alternative approaches on simulated data with $d=2$ targets and heterogeneous Laplace noise. The target average is $90\%$. Results are averaged over $200$ trials. All methods achieves near-nominal coverage and approximate \emph{Pop.~Oracle} as the calibration sample size increases. Our TSCP is sticking above TSCP-LWC with almost indifferent average volume, highlighting its tight containment.}
    \label{fig:approximation1}
\end{figure}

All methods compared in Figure~\ref{fig:approximation1} tend to approximate \emph{Pop.~Oracle} well as $n$ increases, with our TSCP and TSCP-LWC being two most stable and fastest approximation among all methods with theoretical coverage. These results also show that TSCP encloses the aggregated union TSCP-LWC very tightly, and that TSCP-GWC is conservative, though it still performs better than the data-splitting variant TSCP-S. Despite lacking finite-sample coverage guarantee, the \emph{Naive} approach actually performs well, as also shown in Figure~\ref{fig:approximation2}.

\begin{figure}[!htbp]
    \centering
    \includegraphics[width=0.9\textwidth]{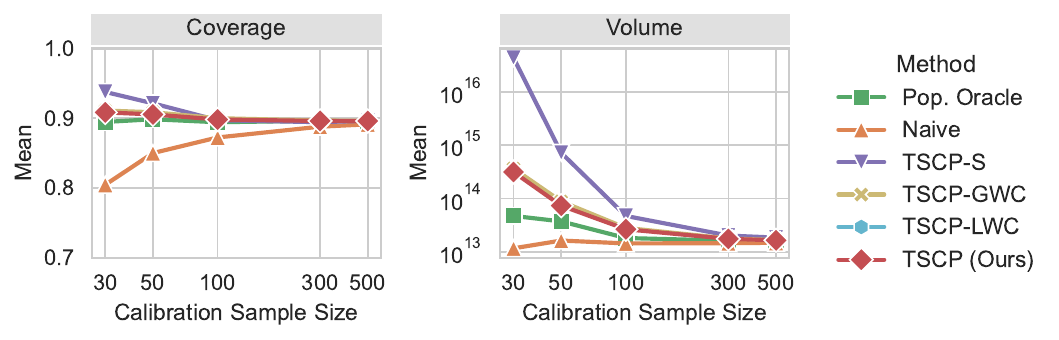}
    \caption{Performance of TSCP and alternative approaches on simulated data with $d=10$ targets. Other details are as in Figure~\ref{fig:approximation1}. TSCP-LWC is omitted from this figure due to its prohibitive cost with even moderately high-dimensional outcomes.}
    \label{fig:approximation2}
\end{figure}

Figure~\ref{fig:runtime} better highlights the prohibitive computational cost of TSCP-LWC, by reporting on the results of similar experiments where a small sample size is fixed, $|\dataset_{\mathrm{cal}}|=30$, while the outcome dimensions are gradually increased: $d \in \{2, 3, 4, 5, 10, 20, 30\}$.

\begin{figure}[!htbp]
    \centering
    \includegraphics[width=\textwidth]{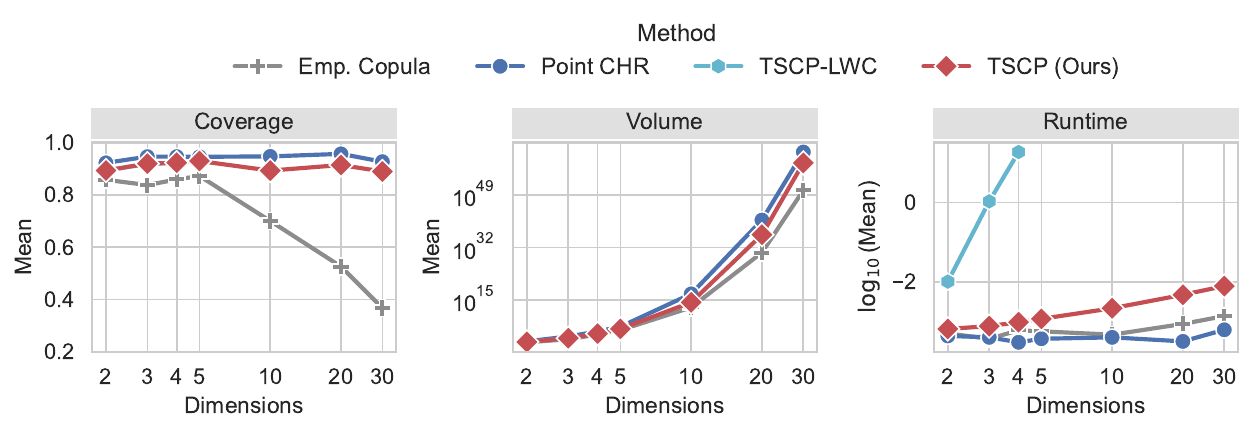}
    \caption{Performance of TSCP and alternative approaches on simulated data with fixed sample size $|\dataset_{\mathrm{cal}}|=30$, as a function of the number of output dimensions $d$. Other details are as in Figure~\ref{fig:approximation1}. These results are averaged over $30$ independent repetitions. The runtime results are presented in seconds on a $\log_{10}$ scale. 
    The runtime of all methods except TSCP-LWC scales well with number of dimensions. TSCP-LWC is only applied when $d \in \{2,3,4\}$ due to the exponential growth of its runtime. }
    \label{fig:runtime}
\end{figure}

\clearpage
\subsection{Robustness to Heavy-Tailed Data} \label{app:hevy_tailed_simulation}

We test the performance of TSCP on simulated heavy-tailed data with $d=10$ outcomes and noise following a Student's-$t$ distribution; i.e., $Y_j \mid X \sim t_{\mathrm{DF}}$, with $\mathrm{DF} \in \{1.5,2,3\}$ degrees of freedom. 
For $1 < \mathrm{DF} \leq 2$, the Student's-$t$ distribution leads to residuals with infinite population variance, violating Assumption~\ref{eq:assumption-scores}. Nonetheless, TSCP continues to achieve the desired coverage level even in this extreme setting, although with some loss of informativeness.

\begin{figure}[!htbp]
     \centering
     \begin{subfigure}[b]{\textwidth}
         \centering
         \includegraphics[width=0.9\textwidth]{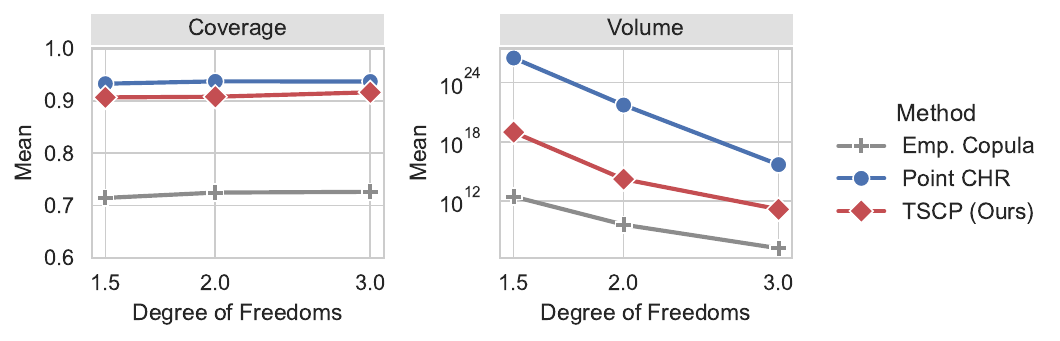}
         \caption{Small sample, $|\dataset_{\mathrm{cal}}|=30$. Output dimensions $d=10$.}
         \label{fig:t-n30-10d}
     \end{subfigure}
     \vfill
     \begin{subfigure}[b]{\textwidth}
         \centering
         \includegraphics[width=0.9\textwidth]{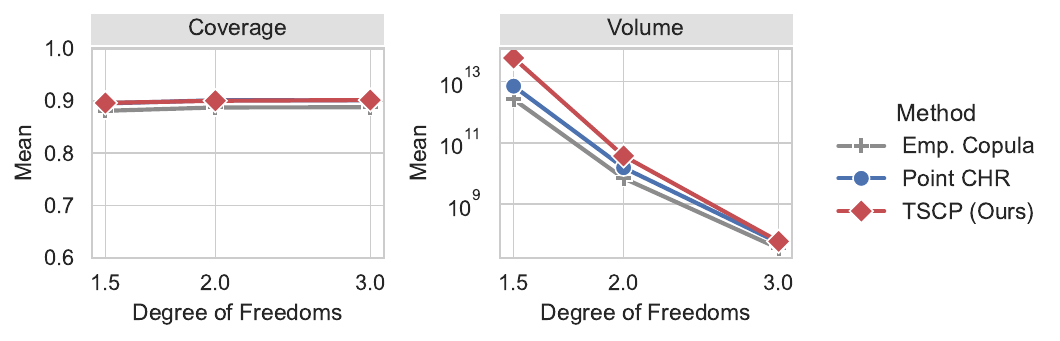}
         \caption{Large sample, $|\dataset_{\mathrm{cal}}|=500$. Output dimensions $d=10$.}
         \label{fig:t-n500-10d}
     \end{subfigure}
        \caption{Performance of the proposed TSCP method and other approaches on simulated data with noise following a Student's-$t$ distribution with varying degree of freedoms, using small (a) and large (b) sample sizes. The target coverage level is 0.9. All methods except \emph{Emp.~Copula} have near-nominal coverage in all settings. TSCP maintains valid marginal coverage even on very heavy-tailed data. In the small sample setting, Point CHR suffers more in terms of prediction set size due to its inefficient use of calibration data, although it can outperform TSCP in large samples.}
        \label{fig:t}
\end{figure}
 



%% file: figures/baselines/gaussian_10d_table.tex
\begin{tabular}{l l c c c}
\toprule
(n, d, Noise) & Method & Coverage & Volume  \\
\midrule
\multirow{4}{*}{(30, 10, Gaussian)} & Emp. Copula & \textcolor{red}{0.714 (0.079)} & 1.68e+10(1.16e+10) \\
 & Point CHR & 0.933 (0.066) & 4.04e+12(1.38e+13) \\
 & TSCP (Ours) & 0.910 (0.053) & \textbf{1.83e+11(2.74e+11)} \\
 & Unscaled Max & 0.908 (0.047) & 2.22e+13(4.77e+13) \\
\midrule
\multirow{4}{*}{(50, 10, Gaussian)} & Emp. Copula & \textcolor{red}{0.816 (0.057)} & 3.64e+10(2.22e+10) \\
 & Point CHR & 0.919 (0.057) & 4.85e+11(1.72e+12) \\
 & TSCP (Ours) & 0.904 (0.047) & \textbf{9.42e+10(8.82e+10)} \\
 & Unscaled Max & 0.897 (0.043) & 1.16e+13(1.10e+13) \\
\midrule
\multirow{4}{*}{(100, 10, Gaussian)} & Emp. Copula & \textcolor{red}{0.819 (0.039)} & 2.58e+10(9.90e+09) \\
 & Point CHR & 0.901 (0.041) & 9.74e+10(9.22e+10) \\
 & TSCP (Ours) & 0.903 (0.034) & \textbf{6.59e+10(3.43e+10)} \\
 & Unscaled Max & 0.908 (0.027) & 1.09e+13(6.93e+12) \\
\midrule
\multirow{4}{*}{(300, 10, Gaussian)} & Emp. Copula & \textcolor{red}{0.872 (0.019)} & 3.73e+10(8.11e+09) \\
 & Point CHR & 0.899 (0.025) & 5.60e+10(2.05e+10) \\
 & TSCP (Ours) & 0.899 (0.018) & \textbf{4.93e+10(1.17e+10)} \\
 & Unscaled Max & 0.901 (0.020) & 8.68e+12(3.29e+12) \\
\midrule
\multirow{4}{*}{(500, 10, Gaussian)} & Emp. Copula & \textcolor{red}{0.886 (0.018)} & 4.16e+10(8.27e+09) \\
 & Point CHR & 0.900 (0.022) & 5.13e+10(1.49e+10) \\
 & TSCP (Ours) & 0.901 (0.016) & \textbf{4.81e+10(9.67e+09)} \\
 & Unscaled Max & 0.903 (0.016) & 8.54e+12(2.46e+12) \\
\midrule
\bottomrule
\end{tabular}

%% file: figures/baselines/unit_gaussian_10d_table.tex
\begin{tabular}{l l c c c}
\toprule
(n, d, Noise) & Method & Coverage & Volume  \\
\midrule
\multirow{4}{*}{(30, 10, Unit Gaussian)} & Emp. Copula & \textcolor{red}{0.714 (0.079)} & 4.64e+03(3.20e+03) \\
 & Point CHR & 0.933 (0.066) & 1.11e+06(3.80e+06) \\
 & TSCP (Ours) & 0.910 (0.053) & 5.04e+04(7.56e+04) \\
 & Unscaled Max & 0.904 (0.054) & \textbf{2.19e+04(1.94e+04)} \\
\midrule
\multirow{4}{*}{(50, 10, Unit Gaussian)} & Emp. Copula & \textcolor{red}{0.816 (0.057)} & 1.00e+04(6.12e+03) \\
 & Point CHR & 0.919 (0.057) & 1.34e+05(4.75e+05) \\
 & TSCP (Ours) & 0.904 (0.047) & 2.60e+04(2.43e+04) \\
 & Unscaled Max & 0.901 (0.045) & \textbf{1.68e+04(1.15e+04)} \\
\midrule
\multirow{4}{*}{(100, 10, Unit Gaussian)} & Emp. Copula & \textcolor{red}{0.819 (0.039)} & 7.11e+03(2.73e+03) \\
 & Point CHR & 0.901 (0.041) & 2.68e+04(2.54e+04) \\
 & TSCP (Ours) & 0.903 (0.034) & 1.82e+04(9.46e+03) \\
 & Unscaled Max & 0.900 (0.035) & \textbf{1.50e+04(7.38e+03)} \\
\midrule
\multirow{4}{*}{(300, 10, Unit Gaussian)} & Emp. Copula & \textcolor{red}{0.872 (0.019)} & 1.03e+04(2.23e+03) \\
 & Point CHR & 0.899 (0.025) & 1.54e+04(5.66e+03) \\
 & TSCP (Ours) & 0.899 (0.018) & 1.36e+04(3.23e+03) \\
 & Unscaled Max & 0.898 (0.018) & \textbf{1.28e+04(3.07e+03)} \\
\midrule
\multirow{4}{*}{(500, 10, Unit Gaussian)} & Emp. Copula & \textcolor{red}{0.886 (0.018)} & 1.15e+04(2.28e+03) \\
 & Point CHR & 0.900 (0.022) & 1.41e+04(4.11e+03) \\
 & TSCP (Ours) & 0.901 (0.016) & 1.33e+04(2.67e+03) \\
 & Unscaled Max & 0.899 (0.017) & \textbf{1.26e+04(2.50e+03)} \\
\midrule
\bottomrule
\end{tabular}

%% file: figures/baselines/laplace_10d_table.tex
\begin{tabular}{l l c c c}
\toprule
(n, d, Noise) & Method & Coverage & Volume  \\
\midrule
\multirow{4}{*}{(30, 10, Laplace)} & Emp. Copula & \textcolor{red}{0.712 (0.079)} & 3.82e+12(4.91e+12) \\
 & Point CHR & 0.938 (0.060) & 5.57e+16(2.71e+17) \\
 & TSCP (Ours) & 0.908 (0.060) & \textbf{3.14e+14(9.93e+14)} \\
 & Unscaled Max & 0.897 (0.053) & 6.36e+15(1.99e+16) \\
\midrule
\multirow{4}{*}{(50, 10, Laplace)} & Emp. Copula & \textcolor{red}{0.819 (0.056)} & 1.36e+13(1.75e+13) \\
 & Point CHR & 0.917 (0.050) & 6.64e+14(1.81e+15) \\
 & TSCP (Ours) & 0.906 (0.045) & \textbf{7.32e+13(1.60e+14)} \\
 & Unscaled Max & 0.901 (0.041) & 3.34e+15(5.12e+15) \\
\midrule
\multirow{4}{*}{(100, 10, Laplace)} & Emp. Copula & \textcolor{red}{0.811 (0.039)} & 5.96e+12(3.94e+12) \\
 & Point CHR & 0.896 (0.039) & 4.74e+13(6.07e+13) \\
 & TSCP (Ours) & 0.898 (0.030) & \textbf{2.64e+13(1.94e+13)} \\
 & Unscaled Max & 0.895 (0.029) & 1.77e+15(1.78e+15) \\
\midrule
\multirow{4}{*}{(300, 10, Laplace)} & Emp. Copula & \textcolor{red}{0.871 (0.020)} & 1.11e+13(4.03e+12) \\
 & Point CHR & 0.895 (0.023) & 2.01e+13(1.12e+13) \\
 & TSCP (Ours) & 0.896 (0.017) & \textbf{1.74e+13(6.40e+12)} \\
 & Unscaled Max & 0.896 (0.017) & 1.38e+15(7.08e+14) \\
\midrule
\multirow{4}{*}{(500, 10, Laplace)} & Emp. Copula & \textcolor{red}{0.882 (0.016)} & 1.29e+13(4.10e+12) \\
 & Point CHR & 0.896 (0.021) & 1.83e+13(8.95e+12) \\
 & TSCP (Ours) & 0.896 (0.016) & \textbf{1.64e+13(5.68e+12)} \\
 & Unscaled Max & 0.897 (0.017) & 1.36e+15(5.51e+14) \\
\midrule
\bottomrule
\end{tabular}

%% file: figures/baselines/mixed_10d_table.tex
\begin{tabular}{l l c c c}
\toprule
(n, d, Noise) & Method & Coverage & Volume  \\
\midrule
\multirow{4}{*}{(30, 10, Mixed)} & Emp. Copula & \textcolor{red}{0.715 (0.075)} & 2.69e+11(3.04e+11) \\
 & Point CHR & 0.935 (0.062) & 2.56e+15(1.98e+16) \\
 & TSCP (Ours) & 0.913 (0.051) & \textbf{9.53e+12(3.12e+13)} \\
 & Unscaled Max & 0.903 (0.052) & 3.46e+13(6.74e+13) \\
\midrule
\multirow{4}{*}{(50, 10, Mixed)} & Emp. Copula & \textcolor{red}{0.821 (0.058)} & 7.07e+11(5.96e+11) \\
 & Point CHR & 0.925 (0.048) & 6.10e+13(3.06e+14) \\
 & TSCP (Ours) & 0.906 (0.044) & \textbf{2.57e+12(3.91e+12)} \\
 & Unscaled Max & 0.901 (0.044) & 1.98e+13(1.68e+13) \\
\midrule
\multirow{4}{*}{(100, 10, Mixed)} & Emp. Copula & \textcolor{red}{0.819 (0.035)} & 3.82e+11(1.98e+11) \\
 & Point CHR & 0.898 (0.042) & 2.76e+12(4.47e+12) \\
 & TSCP (Ours) & 0.900 (0.029) & \textbf{1.23e+12(7.54e+11)} \\
 & Unscaled Max & 0.905 (0.028) & 1.65e+13(9.89e+12) \\
\midrule
\multirow{4}{*}{(300, 10, Mixed)} & Emp. Copula & \textcolor{red}{0.874 (0.022)} & 6.47e+11(2.25e+11) \\
 & Point CHR & 0.902 (0.025) & 1.27e+12(7.23e+11) \\
 & TSCP (Ours) & 0.900 (0.019) & \textbf{9.05e+11(3.48e+11)} \\
 & Unscaled Max & 0.900 (0.020) & 1.36e+13(5.22e+12) \\
\midrule
\multirow{4}{*}{(500, 10, Mixed)} & Emp. Copula & \textcolor{red}{0.886 (0.016)} & 7.26e+11(1.82e+11) \\
 & Point CHR & 0.897 (0.020) & 1.01e+12(5.42e+11) \\
 & TSCP (Ours) & 0.898 (0.014) & \textbf{8.17e+11(2.05e+11)} \\
 & Unscaled Max & 0.902 (0.016) & 1.35e+13(3.73e+12) \\
\midrule
\bottomrule
\end{tabular}

%% file: figures/baselines/gamma_10d_table.tex
\begin{tabular}{l l c c c}
\toprule
(n, d, Noise) & Method & Coverage & Volume  \\
\midrule
\multirow{4}{*}{(30, 10, Gamma)} & Emp. Copula & \textcolor{red}{0.714 (0.067)} & 5.33e-02(6.85e-02) \\
 & Point CHR & 0.929 (0.058) & 7.55e+02(4.32e+03) \\
 & TSCP (Ours) & 0.904 (0.056) & \textbf{3.77e+00(1.59e+01)} \\
 & Unscaled Max & 0.898 (0.051) & 2.69e+15(6.18e+15) \\
\midrule
\multirow{4}{*}{(50, 10, Gamma)} & Emp. Copula & \textcolor{red}{0.826 (0.048)} & 2.20e-01(2.78e-01) \\
 & Point CHR & 0.927 (0.049) & 5.01e+01(3.36e+02) \\
 & TSCP (Ours) & 0.910 (0.041) & \textbf{1.15e+00(2.56e+00)} \\
 & Unscaled Max & 0.902 (0.039) & 1.94e+15(3.54e+15) \\
\midrule
\multirow{4}{*}{(100, 10, Gamma)} & Emp. Copula & \textcolor{red}{0.818 (0.038)} & 1.04e-01(1.01e-01) \\
 & Point CHR & 0.906 (0.038) & 1.28e+00(2.38e+00) \\
 & TSCP (Ours) & 0.900 (0.030) & \textbf{4.65e-01(5.82e-01)} \\
 & Unscaled Max & 0.897 (0.030) & 9.67e+14(7.86e+14) \\
\midrule
\multirow{4}{*}{(300, 10, Gamma)} & Emp. Copula & \textcolor{red}{0.873 (0.019)} & 1.75e-01(1.15e-01) \\
 & Point CHR & 0.895 (0.024) & 3.74e-01(3.70e-01) \\
 & TSCP (Ours) & 0.898 (0.018) & \textbf{2.67e-01(1.69e-01)} \\
 & Unscaled Max & 0.895 (0.018) & 7.04e+14(3.45e+14) \\
\midrule
\multirow{4}{*}{(500, 10, Gamma)} & Emp. Copula & \textcolor{red}{0.885 (0.015)} & 1.89e-01(1.06e-01) \\
 & Point CHR & 0.896 (0.020) & 3.20e-01(2.33e-01) \\
 & TSCP (Ours) & 0.897 (0.016) & \textbf{2.36e-01(1.35e-01)} \\
 & Unscaled Max & 0.893 (0.015) & 6.37e+14(2.53e+14) \\
\midrule
\bottomrule
\end{tabular}